\documentclass[11pt]{article}
\usepackage{monashwp}
\usepackage{tikz,amsthm,booktabs,pmat,tabularx,threeparttable}
\usepackage{graphicx,url,paralist,mathtools,float,setspace,rotating,booktabs,longtable,caption,subcaption,pdflscape,threeparttable,enumerate,xpatch,tabularx}
\usepackage[inline]{enumitem}
\usepackage{tikz}
\usepackage{pmat}

\bibliography{myrefs}
\bibliography{pkgs}
\usepackage[section]{placeins}

\DeclareNameAlias{sortname}{family-given}

\newcommand{\E}{\textnormal{E}}

\newcommand{\tr}{\textnormal{tr}}
\newcommand{\by}{\bm{y}}
\newcommand{\be}{\bm{e}}

\newcommand{\bS}{\bm{S}}

\newcommand{\bG}{\bm{G}}
\newcommand{\bW}{\bm{W}}
\newcommand{\bV}{\bm{V}}
\newcommand{\bJ}{\bm{J}}
\newcommand{\bU}{\bm{U}}
\newcommand{\bY}{\bm{Y}}
\newcommand{\bX}{\bm{X}}

\newcommand{\bOmega}{\bm{\Omega}}
\newcommand{\bI}{\bm{I}}
\newcommand{\bb}{\bm{b}}
\newcommand{\bB}{\bm{B}}
\newcommand{\bbeta}{\bm{\beta}}
\newcommand{\bSigma}{\bm{\Sigma}}

\newcommand*{\minidx}{\operatornamewithlimits{min}\limits}

\definecolor{lblue1}{RGB}{160,160,255}
\definecolor{lblue2}{RGB}{210,210,255}
\definecolor{lblue}{RGB}{100,100,255}

\definecolor{lred1}{RGB}{255,160,160}
\definecolor{lred2}{RGB}{255,210,210}
\definecolor{lred3}{RGB}{255,240,240}

\definecolor{lred}{RGB}{255,100,100}

\definecolor{lpink}{RGB}{183,110,183}

\newtheorem{proposition}{Proposition}
\newtheorem{theorem}{Theorem}

\newtheorem{remark}{Remark}

\usepackage[colorinlistoftodos,prependcaption,textsize=tiny]{todonotes}
\usepackage{xargs}
\newcommandx{\shani}[2][1=]{\todo[linecolor=blue,backgroundcolor=blue!25,bordercolor=blue,#1]{#2}}

\newcommandx{\rob}[2][1=]{\todo[linecolor=red,backgroundcolor=red!25,bordercolor=red,#1]{#2}}

\title{Properties of point forecast reconciliation approaches}
\author{Shanika L Wickramasuriya}

\nojel

\nocover

\addresses{\textbf{Shanika L Wickramasuriya}\\
  Department of Statistics,\\
  University of Auckland,\\ Auckland, New Zealand.\\
  Email: s.wickramasuriya@auckland.ac.nz\\
  ORCID: 0000-0003-2742-5992}

\begin{document}
\titlepage

\begin{abstract}
Point forecast reconciliation of collection of time series with linear aggregation constraints has evolved substantially over the last decade. A few commonly used methods are GLS (generalized least squares), OLS (ordinary least squares), WLS (weighted least squares), and MinT (minimum trace). GLS and MinT have similar mathematical expressions, but they differ by the covariance matrix used. OLS and WLS can be considered as special cases of MinT where they differ by the assumptions made about the structure of the covariance matrix. All these methods ensure that the reconciled forecasts are unbiased, provided that the base forecasts are unbiased. The ERM (empirical risk minimizer) approach was proposed to relax the assumption of unbiasedness.

This paper proves that \begin{inparaenum}[(a)] \item GLS and MinT reduce to the same solution; \item on average, a method similar to ERM (which we refer to as MinT-U) can produce better forecasts than MinT (lowest total mean squared error) which is then followed by OLS and then by base; and \item the mean squared error of each series in the structure for MinT-U is smaller than that for MinT which is then followed by that for either OLS or base forecasts\end{inparaenum}. We show these theoretical results using a set of simulation studies. We also evaluate them using the Australian domestic tourism data set.
\end{abstract}


\newpage

\section{Introduction}
Multivariate time series connected via a set of aggregation constraints is known as a grouped time series. If the grouping of the series leads to a unique structure, we call it a hierarchical time series. For example, the total number of students enrolled into a university for a particular year can be disaggregated by faculty/school, then by the department, down to course level, forms a hierarchical time series. On the other hand, sales data can be disaggregated by the geographic areas and then by the product category or vice-versa. If there is no preference for one disaggregation over the other, we can combine these structures to form a grouped time series. The applications of hierarchical or grouped time series arise in various disciplines: retail \citep{pendal17, karmal19}, energy \citep{jeoetal19, benetal20}, tourism \citep{athetal09, beretal20}, labor market \citep{hynetal16} and economics \citep{athetal19} are among others.

Forecasting these structures are challenging and need careful consideration due to several reasons: \begin{enumerate*}[label=(\alph*)]
	\item forecasts of each series need to be accurate and coherent (i.e., forecasts satisfy the same aggregation constraints as the data) to ensure aligned decision making;
	\item predictive distribution of the coherent forecasts is needed to capture the uncertainty present;
	\item should provide coherent forecasts within a reasonable time for large structures.
\end{enumerate*}

While overcoming the first challenge, there are two commonly used approaches in the literature: bottom-up and top-down. The bottom-up approach forecasts only the most disaggregated series and sums them appropriately to obtain the forecasts for each series in the aggregated levels \citep[see][among others]{orcetal68, dunetal76, shlwol79, pendal17, beretal20}. The top-down approach forecasts the completely aggregated time series and then disaggregates this forecast based on some proportions to form the forecasts for the disaggregated series \citep[see][among others]{grosoh90, athetal09, parnas14}.

\citet{hynetal11} proposed another method based on a linear regression model. This approach forecasts all the series in the structure independently (we refer to these as base forecasts) and model them as the sum of the unknown expectations of the future values of the most disaggregated series and an error term. If the base forecasts are unbiased and the variance covariance matrix of the error is known, then the generalized least squares (GLS) estimator of the expected values of the most disaggregated series gives the minimum variance unbiased estimator. In practice, the variance covariance matrix is not readily available. \citet{hynetal11} overcame this issue by computing the reconciled forecasts using the ordinary least squares (OLS) estimator. \citet{hynetal16} suggested using weighted least squares (WLS) estimator to improve the performance of the reconciled forecasts. Assuming that the base forecasts are unbiased, \citet{wicetal19} proposed another method for forecast reconciliation by minimizing the trace of the reconciled forecast error covariance matrix, which is widely known as the MinT approach. They have illustrated that OLS and WLS are special cases of MinT when assumptions are placed on the variance covariance matrix of the base forecast errors. \citet{panetal20} provided a geometrical interpretation about these reconciliation methods by nesting them within the class of projections.

Relaxing the assumption of unbiasedness of base forecasts, \citet{vancug15} introduced a method called game-theoretically optimal reconciliation. The main idea of this method is to choose a set of reconciled forecasts that guarantees the total weighted quadratic loss of the reconciled forecasts is always smaller than that of the initial forecasts. In general, this method does not have a closed-form solution and is computed as a constrained quadratic programming problem using the general purpose optimization software. Hence this method can be problematic when dealing with large structures that we encounter in practice. \citet{benkoo19} proposed a method that seeks a set of reconciled forecasts with the best trade-off between bias and forecast error variance. They also implemented a regularization method to handle large structures. The simulation and empirical results of this method illustrated that it is a competitive method to the existing methods in the literature. \citet{panetal20} suggested that rather than finding linear mappings that are not projections, it may be appropriate to bias-correct the base forecasts. Through empirical results, they found that even when the form of bias-correction is failed, its impact can be mitigated by forecast reconciliation.

Until recently, there had been no effort devoted to address the second challenge of forecasting hierarchical or grouped time series. \citet{benetal20} proposed a method to compute coherent probabilistic forecasts in a bottom-up fashion using a set of permutations derived from empirical copulas. \citet{gam20} extended the point forecast reconciliation to probabilistic forecast reconciliation using linear transformations. They provided conditions under which the linear transformation is a projection and favored an oblique projection similar to that of the MinT approach. Rather than using existing projection matrices for point forecast reconciliation, \citet{panetal20b} proposed to optimize either the energy or variogram score to find the reconciliation weights.

To the best of our knowledge, the third challenge has been fulfilled by the MinT (OLS or WLS) reconciliation approach. The simulation results of \citet{wic17} observed that a structure with nearly 5.5 million time series could be reconciled in less than 10 minutes.

The main contributions of this paper are to prove that
\begin{enumerate}[label = (\alph*)]
	\item GLS and MinT lead to the same projection matrix, even though the two methods minimize different loss functions and have two similar form analytical expressions involving different covariance matrices;
	\item On average, a method similar to \citet{benkoo19} (which we refer to as MinT-U later on) can produce better forecasts than MinT (lowest total mean squared error (MSE)) which is then followed by OLS and then by base forecasts;
	\item the mean squared error of each series in the structure for MinT-U is smaller than that for MinT, which is then followed by that for either OLS or base forecasts.
\end{enumerate}

An interesting set of properties for probabilistic forecast reconciliation can be derived under the Gaussian framework. The theory, simulation and empirical results are presented in \citet{wic2021b}.

The rest of the paper is structured as follows. Section~\ref{sec:priliminaries} presents the notations and a review of point forecast reconciliation methods. In Section~\ref{sec:newtheory}, we introduce the theoretical derivations of the statements mentioned above. Section~\ref{sec:simulations} and \ref{sec:application} show the results from simulations and Australian domestic tourism data set, respectively. Section~\ref{sec:conclusion} conclude with a short discussion of possible future research.

\section{Preliminaries}
\label{sec:priliminaries}
\subsection{Notation}

Let $\by_t \in \mathbb{R}^m$ be a vector of all observations collected at time $t$ from each series in the structure, and $\bb_t \in \mathbb{R}^n$ be a vector formed only using the observations collected at time $t$ from the most disaggregated level. These are connected via
\begin{align}
	\label{eq:obsstr}
	\by_t = \bS\bb_t,
\end{align}
where $\bS$ is of order $m \times n$ which consists of aggregation constraints present in the structure. To clarify these notations and relationships more clearly, consider the structure given in Figure~\ref{fig:tree}. Let's define a generic series within the structure as $X$, with $y_{X, t}$ denoting the value of series $X$ at time $t$ and $y_t$ being the aggregate of series in the most disaggregated level at time $t$.

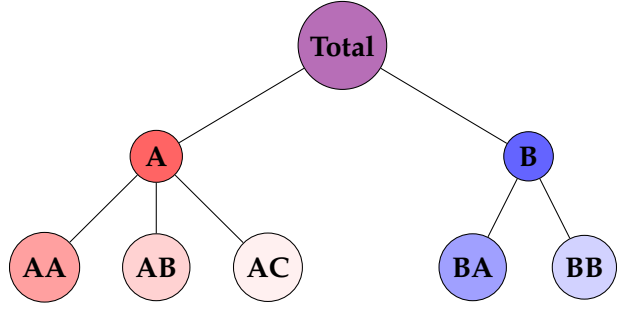
\begin{figure}  \center
	\resizebox{0.5\textwidth}{0.25\textwidth}{%
		\begin{tikzpicture}
			\tikzstyle{every node}=[circle,draw,inner sep=3pt]
			\tikzstyle[level distance=.1cm]
			\tikzstyle[sibling distance=.1cm]
			\tikzstyle{level 1}=[sibling distance=50mm,font=\small]
			\tikzstyle{level 2}=[sibling distance=15mm,font=\footnotesize]
			\node[fill=lpink, font=\bfseries]{Total}
			child {node[fill=lred, font=\bfseries] {A}
				child {node[fill=lred1, font=\bfseries] {AA}}
				child {node[fill=lred2, font=\bfseries] {AB}}
				child {node[fill=lred3, font=\bfseries] {AC}}
			}
			child {node[fill=lblue, font=\bfseries] {B}
				child {node[fill=lblue1, font=\bfseries] {BA}}
				child {node[fill=lblue2, font=\bfseries] {BB}}
			};
	\end{tikzpicture}}
	\caption{An example of a two-level tree.}
	\label{fig:tree}
\end{figure}

For the structure given in Figure~\ref{fig:tree}, $m = 8$, $n = 5$, $\bb_t = [y_{AA, t}, y_{AB, t}, y_{AC, t}, y_{BA, t}, y_{BB, t}]^\top$, $\by_t = [y_t, y_{A, t}, y_{B, t}, y_{AA, t}, y_{AB, t}, y_{AC, t}, y_{BA, t}, y_{BB, t}]^\top$,  and
$$
\bS = \left[\begin{array}{ccccc}
	1 & 1 & 1 & 1 & 1\\
	1 & 1 & 1 & 0 & 0\\
	0 & 0 & 0 & 1 & 1\\
	& & \bI_5 & &
\end{array}
\right],
$$
where $\bI_k$ denotes an identity matrix of order $k \times k$.

These notations can be easily extended to any large collection of time series subject to any aggregation constraints. We should also emphasize that the definition of $\bS$ and $\bb_t$ can differ depending on the application \citep{sha17, jeoetal19}.

Define $\hat{\by}_{t+h|t} \in \mathbb{R}^m$ to be the vector consisting of $h$-step-ahead base forecasts for each time series in the structure, made using observations up to and including time $t$, and arranged in the same order as $\by_t$. The reconciliation methods with linear constraints can be expressed as
$$
\tilde{\by}_{t+h|t} = \bS\bG_h\hat{\by}_{t+h|t},
$$
where $\bG_h \in \mathbb{R}^{n \times m}$ is a matrix which linearly maps a set of base forecasts into a new set of forecasts which are then linearly combined by pre-multiplying with $\bS$ to form a set of reconciled forecasts given by $\tilde{\by}_{t+h|t}$. Assuming that the base forecasts are unbiased, \citet{hynetal11} have shown that the reconciled forecasts are unbiased if and only if $\bS\bG_h\bS = \bS$ or equivalently, $\bG_h\bS = \bI_n$ holds. \citet{panetal20} have provided a different interpretation to these conditions. They have shown that these conditions are equivalent to assuming that $\bS\bG_h$ is a projection matrix onto the column space of $\bS$.

\subsection{Point forecast reconciliation methods}
\label{sec:pfrecon}

\subsubsection{GLS reconciliation}
\citet{hynetal11} considered the following regression model for developing a reconciliation method:
\begin{align}
	\label{eq:reg}
	\hat{\by}_{t+h|t} = \bS\bm{\beta}_{t+h|t} + \bm{\varepsilon}_{t+h},
\end{align}
where $\bm{\beta}_{t+h|t} \in \mathbb{R}^n$ is the vector of unknown means at the most disaggregated level and $\bm{\varepsilon}_{t+h}$ is the coherence error with mean zero and variance covariance matrix $\bm{\Sigma}_h$. These errors are assumed to be independent of observations $\by_1, \by_2, \dots, \by_t$.

If $\bm{\Sigma}_h$ was known, the GLS estimator of $\bm{\beta}_{t+h|t}$ gives the minimum variance unbiased estimator, which results reconciled forecasts to be given by
\begin{align}
	\label{eq:gls}
	\tilde{\by}_{t+h|t} = \bS\left(\bS^\top\bm{\Sigma}_h^\dagger\bS\right)^{-1}\bS^\top\bm{\Sigma}_h^\dagger\hat{\by}_{t+h|t},
\end{align}
where $\bm{\Sigma}_h^\dagger$ is the Moore-Penrose generalized inverse of $\bm{\Sigma}_h$. In general $\bm{\Sigma}_h$ is unknown and \citet{wicetal19} have shown that identifiability issues could arise when the residuals from the regression model in Eq.\ \eqref{eq:reg} are used to estimate $\bm{\Sigma}_h$. \citet{hynetal11} avoided the necessity of this estimate by assuming that $\bm{\Sigma}_h = k_h\bI_m$ for all $h$, where $k_h$ is a positive constant. This is widely known as the OLS approach. \citet{hynetal16} proposed a WLS estimator for $\bm{\beta}_{t+h|t}$ by assuming that $\bm{\Sigma}_h = k_h\bm{\Lambda}_h$ for all $h$, where $\bm{\Lambda}_h$ is a diagonal matrix with elements given by the variances of $\bm{\varepsilon}_{t+h}$ and $k_h > 0$. Due to estimation difficulties with the diagonal elements of $\bm{\Lambda}_h$, they recommended to use the variances of one-step-ahead in-sample base forecast errors.

\subsubsection{MinT reconciliation}

\citet{wicetal19} proposed a reconciliation approach by minimizing the trace of the $h$-step-ahead covariance matrix of the reconciled forecast errors. This method assumes that the base forecasts are unbiased and incorporated $\bG_h\bS = \bI_n$ conditions into the minimization problem to ensure that the reconciled forecasts are also unbiased. Specifically, they considered the following optimization problem:
\begin{align*}
	\minidx_{\bG_h}\E\left\|\by_{t+h} - \tilde{\by}_{t+h|t}\right\|_2^2  = \minidx_{\bG_h}\E\left\|\by_{t+h}\right. & - \left. \bS\bG_h\hat{\by}_{t+h|t}\right\|_2^2 = \minidx_{\bG_h}\tr\left[\bS\bG_h\bW_h\bG_h^\top\bS^\top\right]\\
	\text{s.t.}\quad \bG_h\bS & = \bI_n,
\end{align*}
where $\|\cdot\|_2$ denotes the $l_2$-norm, $\tr(\cdot)$ denotes the trace of a square matrix and $\bW_h$ is the positive definite covariance matrix of the $h$-step-ahead base forecast errors (i.e. $\by_{t+h} - \hat{\by}_{t+h|t}$). The unique solution leads reconciled forecasts to be computed by
\begin{align}
	\label{eq:mint}
	\tilde{\by}_{t+h|t} & = \bS\left(\bS^\top\bW_h^{-1}\bS\right)^{-1}\bS^\top\bW_h^{-1}\hat{\by}_{t+h|t}\\
	& = \bS\left[\bJ - \bJ\bW_h\bU\left(\bU^\top\bW_h\bU\right)^{-1}\bU^\top\right]\hat{\by}_{t+h|t}, \nonumber
\end{align}
where $ \bS^\top = \begin{pmat}[|] \bm{C}^\top & \bI_n\cr\end{pmat}$, ~
$\bm{J} = \begin{pmat}[|] \bm{0}_{n \times m^{*}} & \bI_n\cr\end{pmat}$, ~
$\bm{U}^\top = \begin{pmat}[|]\bI_{m^{*}}  & -\bm{C}\cr\end{pmat}$, ~ and
$m^{*} = m - n$.

Even though $\bG_h$ takes the same form for GLS and MinT, the covariance matrices which enter Eqs.~\eqref{eq:gls} and \eqref{eq:mint} are entirely different: the former is the covariance matrix of the coherence errors, and the latter is the covariance matrix of the base forecast errors.

Assuming that $\bW_h = k_h \bW_1$ for all $h$, where $k_h > 0$ and $h$-step-ahead base forecast errors are jointly covariance stationary, \citet{wicetal19} suggested two estimators for $\bW_h$. They are the unbiased sample covariance matrix and the shrinkage estimator with diagonal target comprising of diagonal elements of the sample covariance matrix, which we refer to as MinT(Sample) and MinT(Shrink), respectively.

\subsubsection{Empirical risk minimizer (ERM) reconciliation}

\citet{benkoo19} introduced a reconciliation approach that does not depend on the unbiasedness of base or reconciled forecasts. Let $\by_1,\by_2, \dots, \by_T$ be observations of all the time series in the structure, and $T$ is the length of each time series. The ERM reconciliation considered the following optimization problem:
\begin{align*}
	\minidx_{\bG_h} \frac{1}{m(T-T_1-h+1)}\left\|\bY_h - \hat{\bY}_h\bG_h^\top\bS^\top \right\|_F^2,
\end{align*}
where $T_1$ is the number of observations used for model fitting, $\|\cdot\|_F$ is the Frobenius norm defined as $\displaystyle \|\bm{X}\|_F = \sqrt{\tr(\bm{X}^\top\bm{X})}$,
\begin{align*}
	\bY_h & = \bB_h\bS^\top = \left[\by_{T_1+h}, \by_{T_1+h+1}, \dots, \by_{T}\right]^\top \in \mathbb{R}^{(T-T_1-h+1) \times m}, \\
	\bB_h & = \left[\bb_{T_1+h}, \bb_{T_1+h+1}, \dots, \bb_{T}\right]^\top \in \mathbb{R}^{(T-T_1-h+1) \times n} \quad \text{and}\\
	\hat{\bY}_h & = \left[\hat{\by}_{T_1+h|T_1}, \hat{\by}_{T_1+h+1|T_1+1}, \dots, \hat{\by}_{T|T-h}\right]^\top \in \mathbb{R}^{(T-T_1-h+1) \times m}.
\end{align*}

Assuming that the observations are jointly weakly stationary and $\hat{\bY}_h^\top\hat{\bY}_h$ is positive definite, the reconciled forecasts from ERM can be computed as
\begin{align*}
	\tilde{\by}_{T+h|T} = \bS\bB_h^\top\hat{\bY}_h\left(\hat{\bY}_h^\top\hat{\bY}_h\right)^{-1}\hat{\by}_{T+h|T}.
\end{align*}

If $\hat{\bY}_h^\top\hat{\bY}_h$ is positive semi-definite, \citet{benkoo19} suggested using the thin singular value decomposition (SVD) of $\hat{\bY}_h$. Then the solution becomes
\begin{align*}
	\tilde{\by}_{T+h|T} = \bS\bB_h^\top\bU_{\hat{\bY}_h}\bm{D}_{\hat{\bY}_h}^{-1}\bV_{\hat{\bY}_h}^\top\hat{\by}_{t+h|t},
\end{align*}
where the thin SVD of $\hat{\bY}_h$ is $\hat{\bY}_h = \bU_{\hat{\bY}_h}\bm{D}_{\hat{\bY}_h}\bV_{\hat{\bY}_h}^\top$, and $\bU_{\hat{\bY}_h}$, $\bV_{\hat{\bY}_h}$ are matrices with orthonormal columns and $\bm{D}_{\hat{\bY}_h}$ is a diagonal matrix with positive entries representing the singular values of $\hat{\bY}_h$.

\section{Theoretical properties of point forecast reconciliation}
\label{sec:newtheory}
\subsection{Relationship between GLS and MinT}

\citet{wicetal19} considered that GLS and MinT are two different approaches to forecast reconciliation. In this section, we prove that both approaches lead to the same $\bG_h$ matrix even though the variance covariance matrices involved are entirely different.

\begin{proposition}
	\label{prop:glsvsmint}
	The reconciled forecasts from the GLS approach proposed by \citet{hynetal11} and the MinT approach proposed by \citet{wicetal19} are equivalent.
\end{proposition}

\begin{proof}
	The mean squared error of $h$-step-ahead base forecasts can be decomposed as
	\begin{align}
		\label{eq:glsvsmint}
		\bW_h & = \E\left[\by_{t+h} - \hat{\by}_{t+h|t}\right]\left[\by_{t+h} - \hat{\by}_{t+h|t}\right]^\top \nonumber\\
		& = \E\left[\by_{t+h} - \bS\bbeta_{t+h|t}  + \bS\bbeta_{t+h|t} - \hat{\by}_{t+h|t}\right]\left[\by_{t+h} - \bS\bbeta_{t+h|t}  + \bS\bbeta_{t+h|t} - \hat{\by}_{t+h|t}\right]^\top \nonumber\\
		& = \E\left[\by_{t+h} - \bS\bbeta_{t+h|t}\right]\left[\by_{t+h} - \bS\bbeta_{t+h|t}\right]^\top + \E\left[\bS\bbeta_{t+h|t} - \hat{\by}_{t+h|t}\right]\left[\bS\bbeta_{t+h|t} - \hat{\by}_{t+h|t}\right]^\top \nonumber\\
		& = \bS \bOmega_{h} \bS^\top + \bSigma_h.
	\end{align}
	Using the law of iterated expectation, we can show that the cross-product term that appears in the simplification mentioned above becomes zero:
	\begin{align*}
		\E\left[\by_{t+h} - \bS\bbeta_{t+h|t}\right] & \left[\bS\bbeta_{t+h|t}  - \hat{\by}_{t+h|t}\right]^\top \\
		& = \E\left\{\E\left[\left(\by_{t+h} - \bS\bbeta_{t+h|t}\right)\left(\bS\bbeta_{t+h|t} - \hat{\by}_{t+h|t}\right)^\top \big| \by_1, \by_2, \dots, \by_t\right]\right\}\\
		& = \E\left[\bS\bbeta_{t+h|t} - \bS\bbeta_{t+h|t}\right]\left[\bS\bbeta_{t+h|t} - \hat{\by}_{t+h|t}\right]^\top \\
		& = \bm{0}.
	\end{align*}
	
	Substituting Eq.\ \eqref{eq:glsvsmint} into the alternative representation of the MinT approach:
	\begin{align*}
		\bG^{MinT}_h & = \bJ - \bJ\bW_h\bU\left(\bU^\top\bW_h\bU\right)^{-1}\bU^\top \\
		& = \bJ - \bJ\bSigma_h\bU\left(\bU^\top\bSigma_h\bU\right)^{-1}\bU^\top\\
		& = \bG^{GLS}_h,
	\end{align*}
	where $\bG^{GLS}_h$ and $\bG^{MinT}_h$ denote the $\bG_h$ matrix for the GLS and MinT approaches, respectively. The second equality follows from the fact that the columns of the matrix $\bU$ consists of the null-space basis vectors of $\bS^\top$.
\end{proof}

One major implication of Proposition~\ref{prop:glsvsmint} is that if one wishes to use $\bG^{OLS}$ then it may not necessarily indicate that an assumption is made about the covariance matrix of the base forecast errors been isotropic. Alternatively, it may be that the covariance matrix of coherence errors is isotropic, which may be realistic if the $h$-step-ahead base forecasts do not deviate much from the true means (i.e., $\bS\bbeta_{t+h|t}$).

\subsection{Mean squared error bounds for MinT, OLS and base forecasts}

\citet{wicetal19} have argued that the MinT reconciled forecasts are at least as good as the base forecasts by showing that
\begin{align*}
	\left[\by_{t+h} - \tilde{\by}^{MinT}_{t+h|t}\right]^\top\bW_h^{-1}\left[\by_{t+h} - \tilde{\by}^{MinT}_{t+h|t}\right] \leq \left[\by_{t+h} - \hat{\by}_{t+h|t}\right]^\top\bW_h^{-1}\left[\by_{t+h} - \hat{\by}_{t+h|t}\right],
\end{align*}
where $\tilde{\by}^{MinT}_{t+h|t}$ is the MinT reconciled forecasts.

As pointed by \citet{panetal20}, this holds true only if we use a loss function that depends on $\bW_h$. If one wishes to use Euclidean distance (i.e., $\bW_h = \bI_m$), then there can be realizations where the MinT approach performs poorly relative to the base forecasts. In such a case, OLS reconciled forecasts are at least as good as the base forecasts. This can be clearly seen from the fact that
\begin{align*}
	\left\|\by_{t+h} - \tilde{\by}_{t+h} \right\|^2_2 = \left\|\bS\bG_h\left(\by_{t+h} - \hat{\by}_{t+h}\right)\right\|_2^2 & \leq   \left\|\bS\bG_h\right\|_2^2\left\|\by_{t+h} - \hat{\by}_{t+h}\right\|_2^2 \\ 
	& = \sigma^2_{\max}\left\|\by_{t+h} - \hat{\by}_{t+h}\right\|_2^2
\end{align*}
for any $\bG_h$ matrix which satisfies $\bG_h\bS = \bI_n$. $\sigma_{\max}$ denotes the largest singular value of $\bS\bG_h$. It is well known that $\sigma_{\max} \geq 1$ as $\bG_h\bS = \bI_n$ implies that $\bS\bG_h$ is a projection, and the equality holds only for an orthogonal projection (i.e., for the OLS reconciliation approach).

We show in Theorem~\ref{thm:bounds} that \textit{on average} the MinT approach produces the best reconciled forecasts. Hence mis-specifying $\bW_h$ or $\bSigma_h$ as isotropic when it is not can degrade the accuracy of the reconciled forecasts.
\begin{theorem}
	\label{thm:bounds}
	On average, the MinT reconciled forecasts are at least as good (i.e., lowest total MSE) as the base forecasts. In other words,
	\begin{align*}
		\tr\left(\E\left[\by_{t+h} - \tilde{\by}^{MinT}_{t+h|t}\right]\left[\by_{t+h} - \tilde{\by}^{MinT}_{t+h|t}\right]^\top\right) & \leq \tr\left(\E\left[\by_{t+h} - \tilde{\by}^{OLS}_{t+h|t}\right]\left[\by_{t+h} - \tilde{\by}^{OLS}_{t+h|t}\right]^\top\right)\\
		& < \tr\left(\E\left[\by_{t+h} - \hat{\by}_{t+h|t}\right]\left[\by_{t+h} - \hat{\by}_{t+h|t}\right]^\top\right),
	\end{align*}
	where $\tilde{\by}^{OLS}_{t+h|t}$ is the OLS reconciled forecasts. Furthermore, the mean squared reconciled forecast error from the MinT approach for each series in the structure is lower than that of OLS and base forecasts.
\end{theorem}

\begin{proof}
	We first show that
	$$\E\left[\by_{t+h} - \tilde{\by}^{OLS}_{t+h|t}\right]\left[\by_{t+h} - \tilde{\by}^{OLS}_{t+h|t}\right]^\top - \E\left[\by_{t+h} - \tilde{\by}^{MinT}_{t+h|t}\right]\left[\by_{t+h} - \tilde{\by}^{MinT}_{t+h|t}\right]^\top
	$$
	is positive semi-definite.
	\begin{align*}
		\E\left[\by_{t+h} - \tilde{\by}^{OLS}_{t+h|t}\right]\left[\by_{t+h} - \tilde{\by}^{OLS}_{t+h|t}\right]^\top & - \E\left[\by_{t+h} - \tilde{\by}^{MinT}_{t+h|t}\right]\left[\by_{t+h} - \tilde{\by}^{MinT}_{t+h|t}\right]^\top \\
		& = \bS\left(\bS^\top\bS\right)^{-1}\bS^\top\bW_h\bS\left(\bS^\top\bS\right)^{-1}\bS^\top - \\
		& \qquad \bS\left(\bS^\top\bW_h^{-1}\bS\right)^{-1}\bS^\top \\
		& = \bS\bm{D}_h^\top \bW_h\bm{D}_h\bS^\top,
	\end{align*}
	where $\bm{D}_h = \bS\left(\bS^\top\bS\right)^{-1} - \bW_h^{-1}\bS\left(\bS^\top\bW_h^{-1}\bS\right)^{-1}$. The positive definiteness of $\bW_h$ implies that $\bS\bm{D}_h^\top \bW_h\bm{D}_h\bS^\top$ is positive semi-definite. Hence,
	\begin{align}
		\label{eq:mintvsols}
		\tr\left(\E\left[\by_{t+h} - \tilde{\by}^{MinT}_{t+h|t}\right]\left[\by_{t+h} - \tilde{\by}^{MinT}_{t+h|t}\right]^\top\right) \leq \tr\left(\E\left[\by_{t+h} - \tilde{\by}^{OLS}_{t+h|t}\right]\left[\by_{t+h} - \tilde{\by}^{OLS}_{t+h|t}\right]^\top\right).
	\end{align}
	
	Let $(\cdot)_{ii}$ denotes the $i$-th diagonal element of a square matrix. It can be clearly seen that $\left(\bS\bm{D}_h^\top \bW_h\bm{D}_h\bS^\top\right)_{ii} \geq 0$ for $i = 1, 2, \dots, m$ as $\bS\bm{D}_h^\top \bW_h\bm{D}_h\bS^\top$ is positive semi-definite. Therefore,
	\begin{align}
		\label{eq:unimintvsols}
		\left(\E\left[\by_{t+h} - \tilde{\by}^{MinT}_{t+h|t}\right]\left[\by_{t+h} - \tilde{\by}^{MinT}_{t+h|t}\right]^\top\right)_{ii} \leq \left(\E\left[\by_{t+h} - \tilde{\by}^{OLS}_{t+h|t}\right]\left[\by_{t+h} - \tilde{\by}^{OLS}_{t+h|t}\right]^\top\right)_{ii}
	\end{align}
	for $i = 1, 2, \dots, m$.
	
	Similarly, we can show that
	\begin{align}
		\label{eq:olsvsbase}
		\tr\left(\E\left[\by_{t+h} - \tilde{\by}_{t+h|t}\right]\right. & \left.\left[\by_{t+h} - \tilde{\by}_{t+h|t}\right]^\top\right)  - \tr\left(\E\left[\by_{t+h} - \tilde{\by}^{OLS}_{t+h|t}\right]\left[\by_{t+h} - \tilde{\by}^{OLS}_{t+h|t}\right]^\top\right) \nonumber\\
		& = \tr\left(\bW_h\right) - \tr\left(\bS\left(\bS^\top \bS\right)^{-1}\bS^\top\bW_h\bS\left(\bS^\top\bS\right)^{-1}\bS^\top\right) \nonumber\\
		& = \tr\left(\bW_h\right) - \tr\left(\bS^\top\bW_h\bS\left(\bS^\top\bS\right)^{-1}\right) \nonumber\\
		& = \tr\left(\bW_h\left(\bI_m - \bS\left(\bS^\top\bS\right)^{-1}\bS^\top\right)\right) \nonumber\\
		& > 0.
	\end{align}
	The last inequality follows from the fact that $\bI_m - \bS\left(\bS^\top\bS\right)^{-1}\bS^\top$ is the orthogonal projection onto the orthogonal complement of $\bS$ and therefore positive semi-definite, and $\bW_h$ is positive definite.
	
	%
	
	It remains to show that the mean squared reconciled forecast error of MinT is smaller than that of the base forecasts for each series. Consider the difference between the variance covariance matrix of the base forecast errors and that of MinT:
	\begin{align*}
		\bW_h - \bS\left(\bS^\top\bW_h^{-1}\bS)^{-1}\bS^\top\right) & = \bW_h^{1/2}\left[\bI_m - \bW_h^{-1/2}\bS\left(\bS^\top\bW_h^{-1}\bS\right)^{-1}\bS^\top\bW^{-1/2}\right]\bW_h^{1/2}\\
		& = \bW_h^{1/2}\left[\bI_m - \bm{A}_h\left(\bm{A}_h^\top\bm{A}_h\right)^{-1}\bm{A}^\top\right]\bW_h^{1/2},
	\end{align*}
	where $\bm{A}_h = \bW_h^{-1/2}\bS$ and $\bI_m - \bm{A}_h\left(\bm{A}_h^\top\bm{A}_h\right)^{-1}\bm{A}^\top$ is the orthogonal projection onto the orthogonal complement of $\bm{A}_h$ and therefore is positive semi-definite. These facts imply that the difference between two variance covariance matrices is positive semi-definite. Therefore,
	\begin{align}
		\label{eq:unimintvsbase}
		\left(\E\left[\by_{t+h} - \tilde{\by}^{MinT}_{t+h|t}\right]\left[\by_{t+h} - \tilde{\by}^{MinT}_{t+h|t}\right]^\top\right)_{ii} \leq \left(\E\left[\by_{t+h} - \tilde{\by}_{t+h|t}\right]\left[\by_{t+h} - \tilde{\by}_{t+h|t}\right]^\top\right)_{ii}
	\end{align}
	for $i = 1, 2, \dots, m$.
	Eqs.\ \eqref{eq:mintvsols}--\eqref{eq:unimintvsbase} complete the proof.
\end{proof}

\begin{remark}
	An implication of Theorem~\ref{thm:bounds} is that for a given estimate of $\bW_h$, the reconciled forecast error variance of a series in the structure is higher for OLS than MinT. Hence, the width of the point-wise prediction intervals is wider for the OLS approach than MinT under the Gaussian assumption.
\end{remark}

\subsection{Unconstrained MinT reconciliation}

Both \citet{hynetal11} and \citet{wicetal19} considered $\bG_h\bS = \bI_n$ as a set of conditions for preserving the unbiasedness of the reconciled forecasts given that the base forecasts are unbiased. In the geometric interpretation of the forecast reconciliation approaches, \citet{panetal20} showed that these conditions lead $\bS\bG_h$ to be a projection matrix onto the column space of $\bS$. In this section, we study the impact of relaxing these conditions on the accuracy of the reconciled forecasts.

Define $\bG_h = \bJ + \bX_{h}$. The covariance matrix of the reconciled forecast errors can be written as
\begin{align*}
	\E\left[\by_{t+h} - \tilde{\by}_{t+h|t}\right] & \left[\by_{t+h} - \tilde{\by}_{t+h|t}\right]^\top = \bS \E\left[\bb_{t+h} - \bG_h\hat{\by}_{t+h|t}\right]\left[\bb_{t+h} - \bG_h\hat{\by}_{t+h|t}\right]^\top\bS^\top\\
	& = \bS \E\left[\bb_{t+h} - \left(\bJ + \bX_{h}\right)\hat{\by}_{t+h|t}\right]\left[\bb_{t+h} - \left(\bJ + \bX_{h}\right)\hat{\by}_{t+h|t}\right]^\top\bS^\top\\
	& = \bS \E\left[\hat{\be}_{B, t+h|t} - \bX_{h}\hat{\by}_{t+h|t}\right]\left[\hat{\be}_{B, t+h|t} - \bX_{h}\hat{\by}_{t+h|t}\right]^\top\bS^\top,
\end{align*}
where $\hat{\be}_{B, t+h|t} = \bb_{t+h} - \hat{\bb}_{t+h|t}$ is the base forecast error at the bottom level.

Assuming that $\E\left[\hat{\by}_{t+h|t}\hat{\by}_{t+h|t}^\top\right]$ is positive definite, the minimizer of the trace of the covariance matrix of the reconciled errors defined above is given by
\begin{align}
	\label{eq:mintu}
	\bX_{h}^* & = \E\left[\hat{\be}_{B, t+h|t}\hat{\by}_{t+h|t}^\top\right]\E\left[\hat{\by}_{t+h|t}\hat{\by}_{t+h|t}^\top\right]^{-1} \quad \text{and} \nonumber\\
	\bG_h^* & = \bJ + \E\left[\hat{\be}_{B, t+h|t}\hat{\by}_{t+h|t}^\top\right]\E\left[\hat{\by}_{t+h|t}\hat{\by}_{t+h|t}^\top\right]^{-1} \nonumber\\
	& = \bJ\E\left[\by_{t+h}\hat{\by}_{t+h|t}^\top\right]\E\left[\hat{\by}_{t+h|t}\hat{\by}_{t+h|t}^\top\right]^{-1} \nonumber\\
	& = \E\left[\bb_{t+h}\hat{\by}_{t+h|t}^\top\right]\E\left[\hat{\by}_{t+h|t}\hat{\by}_{t+h|t}^\top\right]^{-1}.
\end{align}
We refer to this reconciliation approach as ``MinT-U'' (MinT unconstrained).

\begin{theorem}
	\label{thm:mintuvsmint}
	The total mean squared reconciled forecast error of MinT-U is smaller than that of MinT. In addition, the mean squared reconciled forecast error from MinT-U for each series in the structure is lower than that of MinT.
\end{theorem}

\begin{proof}
	Let $\bV_h = \E[\hat{\by}_{t+h|t}\hat{\by}_{t+h|t}^\top]$. Define the difference in MSEs between MinT-U and MinT as
	\begin{align*}
		\bm{\Delta}_h & = \bS\E\left[\hat{\be}_{B, t+h|t}\hat{\by}_{t+h|t}^\top\right]\left[\bU(\bU^\top\bV_h\bU)^{-1}\bU^\top - \bV_h^{-1}\right]\E\left[\hat{\by}_{t+h|t}\hat{\be}_{B, t+h|t}^\top\right]\bS^\top\\
		& = \bS\E\left[\hat{\be}_{B, t+h|t}\hat{\by}_{t+h|t}^\top\right] \bm{\Delta}_h^*\E\left[\hat{\by}_{t+h|t}\hat{\be}_{B, t+h|t}^\top\right]\bS^\top,
	\end{align*}
	where $\bm{\Delta}_h^* = \bU\left(\bU^\top\bV_h\bU\right)^{-1}\bU^\top - \bV_h^{-1}$.
	
	We need to show that $\bm{\Delta}_h$ (or $\bm{\Delta}_h^*$) is negative semi-definite. Let $\bm{\Omega}_h = \bV_h^{1/2}\bU$.
	\begin{align*}
		\bm{\Delta}_h^* & = - \left[\bV_h^{-1} - \bU\left(\bU^\top\bV_h\bU\right)^{-1}\bU^\top\right]\\
		& = - \bm{\Omega}_h^{-1/2}\left[\bI_m - \bm{\Omega}_h^{1/2}\bU\left(\bU^\top\bV_h\bU\right)^{-1}\bU^\top\bm{\Omega}_h^{1/2}\right]\bm{\Omega}_h^{-1/2} \\
		& = - \bm{\Omega}_h^{-1/2}\left[\bI_m - \bm{A}_h\left(\bm{A}_h^\top\bm{A}_h\right)^{-1}\bm{A}_h^\top\right]\bm{\Omega}_h^{-1/2},
	\end{align*}
	where $\bm{A}_h = \bm{\Omega}_h^{1/2}\bU$. We know that $\bI_n - \bm{A}_h\left(\bm{A}_h^\top\bm{A}_h\right)^{-1}\bm{A}_h^\top$ is the orthogonal projection onto the orthogonal complement of $\bm{A}_h$. Hence it is positive semi-definite. This implies that $\bm{\Delta}_h^*$ is negative semi-definite and $(\bm{\Delta}_h)_{ii} \leq 0$ for $i =  1, 2, \dots, m$.
\end{proof}

\citet{panetal20} argued that relaxing $\bG_h\bS = \bI_n$ may not be desirable to deal with biased forecasts because it can compromise the attractive properties of projections, such as Euclidean or generalized Euclidean distance reducing property. However, Theorem~\ref{thm:mintuvsmint} indicates that there can be gains in forecast accuracy on average by relaxing $\bG_h\bS = \bI_n$ conditions.

Assuming that the series in the structure are jointly weakly stationary, we can estimate the unknown quantities in Eq.\ \eqref{eq:mintu} from their sample counterparts:
\begin{align}
	\label{eq:estmintu}
	\hat{\bG}_h & = \bm{B}_h^\top\hat{\bY}_h\left[\hat{\bY}_h^\top\hat{\bY}_h\right]^{-1},
\end{align}
where
\begin{align*}
	\bY_h & = \left[\by_{h}, \by_{h+1}, \dots, \by_{T}\right]^\top \in \mathbb{R}^{(T-h+1) \times m}, \\
	\bB_h & = \left[\bb_{h}, \bb_{h+1}, \dots, \bb_{T}\right]^\top \in \mathbb{R}^{(T-h+1) \times n}\quad \text{and}\\
	\hat{\bY}_h & = \left[\hat{\by}_{h|0}, \hat{\by}_{h+1|1}, \dots, \hat{\by}_{T|T-h}\right]^\top \in \mathbb{R}^{(T-h+1) \times m}.
\end{align*}
We refer to this method as ``EMinT-U'' (empirical MinT unconstrained). This method is closely related to the ERM reconciliation approach introduced by \citet{benkoo19}. ERM uses a holdout validation set to define $\bY_h, \bB_h$ and $\hat{\bY}_h$, whereas Eq.\ \eqref{eq:estmintu} uses in-sample observations and fitted values.

\begin{remark}
	\label{rmk:mintuvsmint}
	Similarly to Theorem~\ref{thm:mintuvsmint}, we can also show that the total in-sample fit of EMinT-U is always smaller than that of MinT(Sample) for at least $h = 1$.
\end{remark}

\section{Simulations}
\label{sec:simulations}
We perform two simulation designs to evaluate the performance of EMinT-U with the state-of-the-art methods discussed in Section~\ref{sec:pfrecon}. The impact of contemporaneous error correlation at the bottom level on the reconciled forecasts is assessed on a small and a large hierarchical structure in Sections~\ref{sec:simcor} and \ref{sec:simcorlarge}, respectively.

\subsection{Exploring the effect of correlation}
\label{sec:simcor}
We consider a hierarchy with two levels of aggregation and seven series in total. Specifically, four series at the bottom level were aggregated in groups of size two, which were then aggregated to give the most aggregated series. The assumed data generating process for the bottom-level series is a stationary first-order vector autoregressive (i.e., VAR(1)) process:
\begin{align*}
	\bm{b}_t & =
	\begin{bmatrix}
		\bm{A}_1 & \bm{0}\\
		\bm{0} & \bm{A}_2
	\end{bmatrix} \bm{b}_{t-1} + \bm{\varepsilon}_t,
\end{align*}
where $\bm{A}_1$ and $\bm{A}_2$ are $2 \times 2$ matrices with eigenvalues $z_{1,2} = 0.6[\cos(\pi/3) \pm i \sin(\pi/3)]$ and  $z_{3, 4} = 0.9[\cos(\pi/6) \pm i \sin(\pi/6)]$, respectively. We also assumed that $\bm{\varepsilon}_t \sim \mathcal{N}(\bm{0}, \bm{\Sigma})$, where $$\bm{\Sigma} = \begin{bmatrix} \bm{\Sigma}_1 & \bm{0} \\
	\bm{0} & \bm{\Sigma}_1\end{bmatrix}, \quad \text{and} \quad \bm{\Sigma}_1 = \begin{bmatrix}
	2 & \sqrt{6}\rho\\
	\sqrt{6}\rho & 3
\end{bmatrix},$$
and $\rho \in {0, \pm 0.1, \pm 0.2, \pm 0.3, \dots, \pm 0.8}$.

For each series at the bottom level, we generated $T = 101$ or 501 observations, with the last observation being withheld as the test set. Using the remaining observations as the training set, base forecasts are then generated from the best fitted ARMA (autoregressive moving average) models obtained by minimizing the AICc (corrected Akaike information criterion). We used the default settings in the automated algorithm of \citet{hynkha08} which is implemented in the \texttt{forecast} package for R \citep{hynetalp2020}. The base forecasts are then reconciled using the approaches discussed in Section~\ref{sec:pfrecon}. The whole process is repeated 1000 times. In order to provide a comprehensive analysis, we report the results for bottom-up (BU) forecasts.

We have also considered $T = 101, 301$, and real-roots for the matrices $\bm{A}_1$ and $\bm{A}_2$ for this simulation design. However, to save space, we do not present all the results in this paper. The omitted results follow a similar pattern and are available upon request.

The left panels of Figures~\ref{fig:simcor-1} and \ref{fig:simcor-2} show the percentage relative improvements in MSE for the in-sample reconciled forecasts relative to that for 1-step-ahead fitted values, and the right panels show the improvements from the out-of-sample reconciled forecasts relative to the base forecasts when $T = 101$ and $T = 501$, respectively. A negative (positive) value shows that the MSE of the reconciled forecasts is lower (higher) than that of fitted/base values/forecasts.

\begin{figure}[!htp]
	\includegraphics[width=\textwidth]{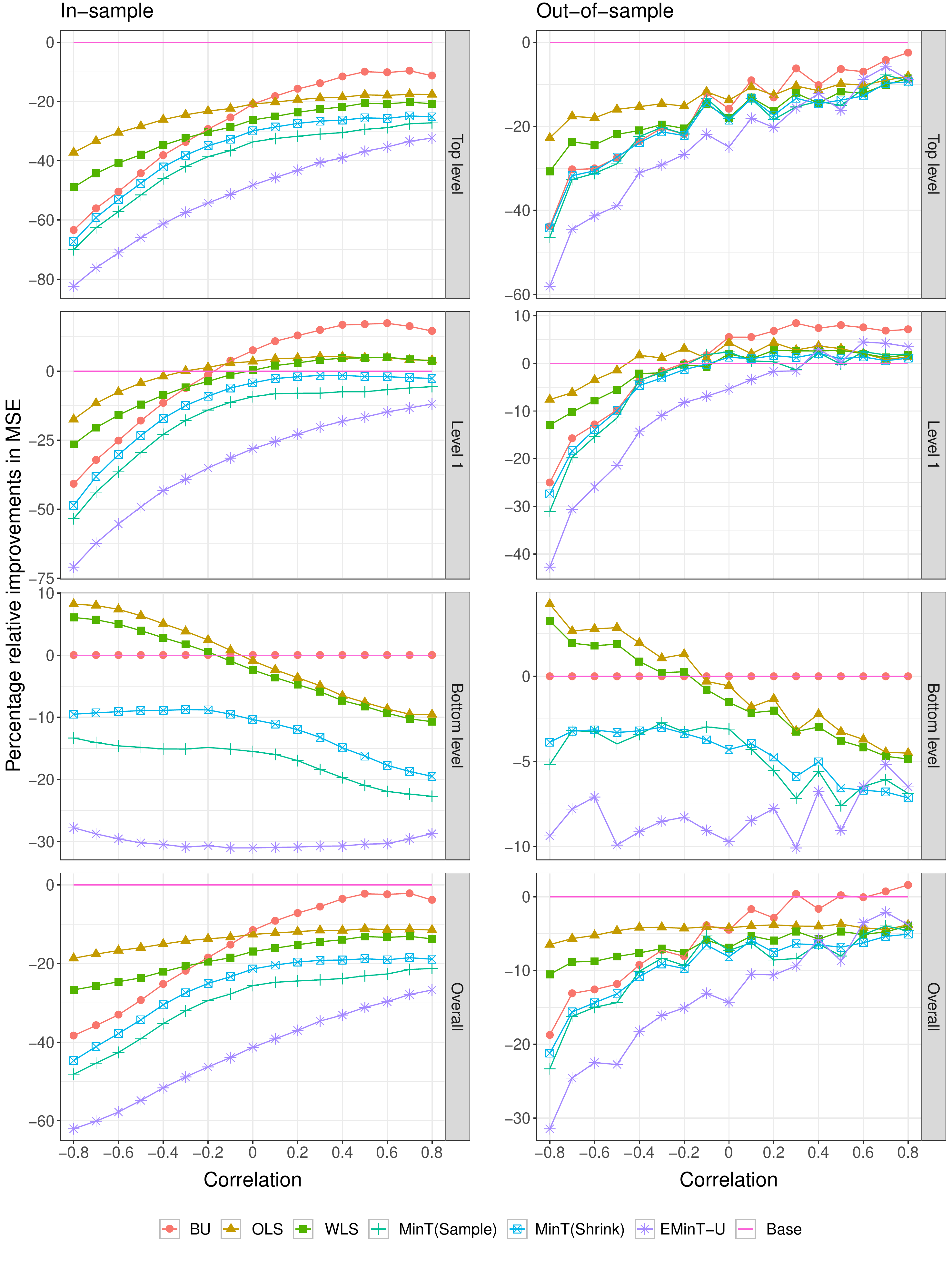}
	\caption{Percentage relative improvements in MSE for the in-sample (shown in the left panel) and out-of-sample (shown in the right panel) evaluations. First three panels show the results for the top level, level 1 and bottom level. The lower panel shows the results for the whole structure. The sample size $T = 101$.}
	\label{fig:simcor-1}
\end{figure}

From Figure~\ref{fig:simcor-1}, it can be clearly seen that the in-sample performance of EMinT-U is greater than that of MinT(Sample) and OLS for all the levels in the structure and error correlations. This is in accordance with Theorem~\ref{thm:bounds} and \ref{thm:mintuvsmint} (where the expectations are replaced by in-sample counterparts), and Remark~\ref{rmk:mintuvsmint}. Specifically for EMinT-U and MinT(Sample), the percentage relative improvements are always negative in each panel, whereas for OLS, it holds only in the last panel. The in-sample MSE of MinT(Sample) is smaller than that of MinT(Shrink). This result is intuitive as MinT(Shrink) uses a shrinkage estimator of the sample covariance matrix. Among the reconciliation methods which use a diagonal covariance matrix, WLS performs better than OLS. This observation is also apparent because the series may have different forecast error variances, and accounting for this in reconciliation could be beneficial. The bottom-up method deteriorates the performances than WLS and OLS for positive error correlations. In general, the overall improvements in all forecast reconciliation approaches have decreased as the error correlation increases from $-0.8$ to 0.8.

\begin{figure}[!htp]
	\includegraphics[width=\textwidth]{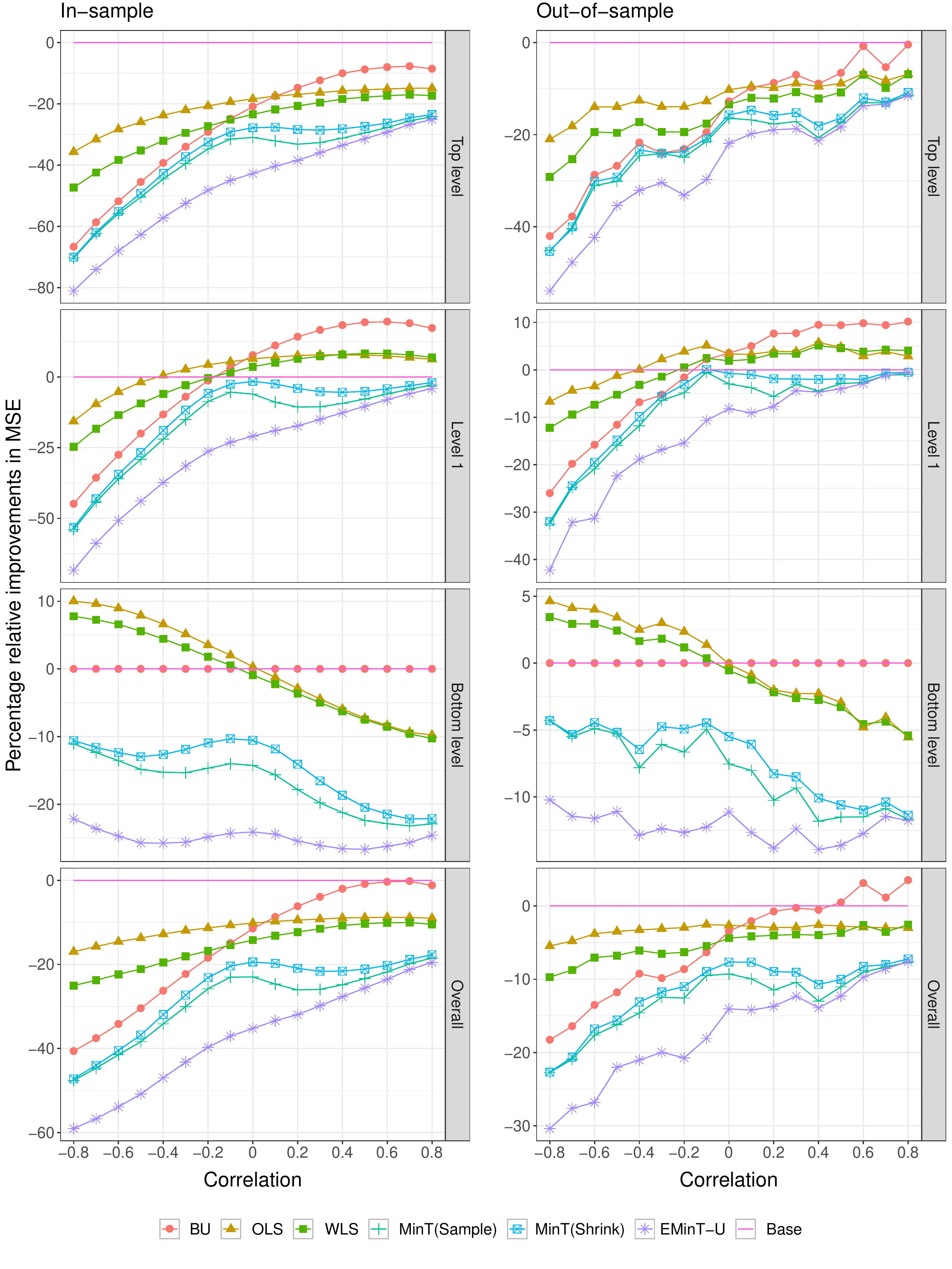}
	\caption{Percentage relative improvements in MSE for the in-sample (shown in the left panel) and out-of-sample (shown in the right panel) evaluations. First three panels show the results for the top level, level 1 and bottom level. The lower panel shows the results for the whole structure. The sample size $T = 501$.}
	\label{fig:simcor-2}
\end{figure}

Some patterns we observed in in-sample evaluations are also present in out-of-sample evaluations for negative error correlations. However, the ordering of performances for EMinT-U, MinT(Sample) and MinT(Shrink) have twisted slightly for positive error correlations. Figure~\ref{fig:simcor-2} shows a similar performance comparison for a larger training set. We can observe that all the patterns we noted previously in in-sample evaluations are extended to out-of-sample evaluations.

\begin{figure}[!htp]
	\includegraphics[width=\textwidth]{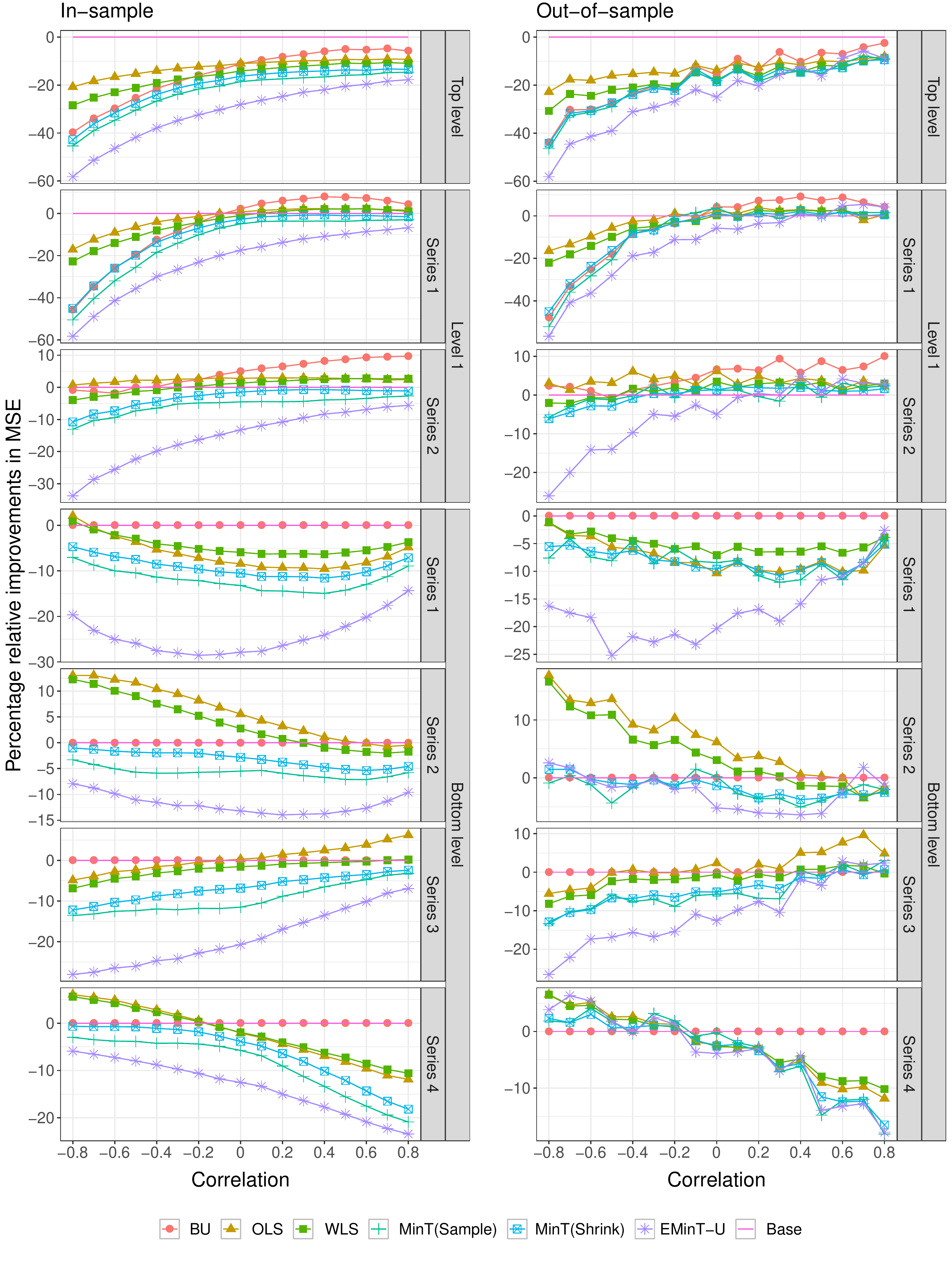}
	\caption{Percentage relative improvements in MSE for the in-sample (shown in the left panel) and out-of-sample (shown in the right panel) evaluations of each series in the structure. The panels from top to bottom are organized based on the level to which each series belongs. The sample size $T = 101$.}
	\label{fig:simcorind-1}
\end{figure}

\begin{figure}[!htp]
	\includegraphics[width=\textwidth]{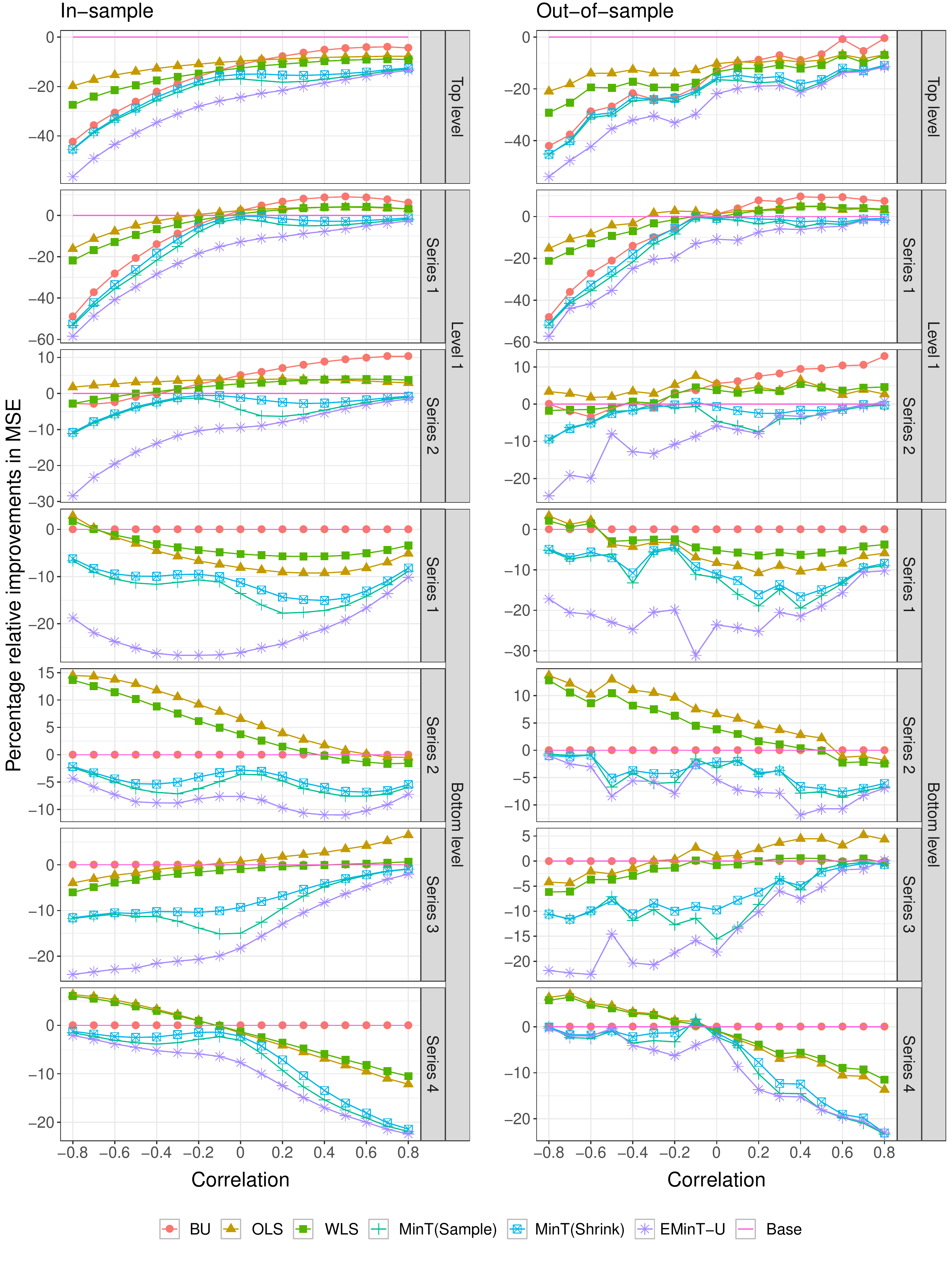}
	\caption{Percentage relative improvements in MSE for the in-sample (shown in the left panel) and out-of-sample (shown in the right panel) evaluations of each series in the structure. The panels from top to bottom are organized based on the level to which each series belongs. The sample size $T = 501$.}
	\label{fig:simcorind-2}
\end{figure}

The left panels of Figures~\ref{fig:simcorind-1} and \ref{fig:simcorind-2} depict the percentage relative improvements in MSE for the in-sample reconciled forecasts relative to that for 1-step-ahead fitted values, and the right panels depict the improvements from the out-of-sample reconciled forecasts relative to the base forecasts for each series in the structure. A negative (positive) value shows that the MSE of the reconciled forecasts is lower (higher) than that of fitted/base values/forecasts. It can be observed that the in-sample performance obeys the following order: EMinT-U, MinT(Sample)/MinT(Shrink) and OLS (arranged in descending order of performance) in each panel. For EMinT-U and MinT(Sample)/MinT(Shrink), the improvement is always larger than the base forecasts. These observations are in accordance with the findings in Theorem~\ref{thm:bounds} and \ref{thm:mintuvsmint}. The out-of-sample performances for $T= 101$ follow similar patterns as noted above, with exceptions occurring at few bottom-level series or/and when the error correlation is positive. However, the patterns become more apparent when the sample size increases to $T = 501$.

\subsection{Exploring the effect of correlation on a larger hierarchy}
\label{sec:simcorlarge}

In this simulation design, we consider a slightly larger hierarchy. The structure consists of two-levels and 43 series in total. There are 36 series at the bottom level and are aggregated in groups of size six to form six series at level 1, which are then aggregated to form the total series. Similarly to the previous simulation setup, we assumed a VAR(1) process to generate the observations at the bottom level. The coefficient matrix of the VAR(1) process is shown in the first panel of Figure~\ref{fig:coefandcor}.

We considered two representations for the correlation matrix of the Gaussian innovation process: \begin{inparaenum}[(a)] \item all the correlations are non-negative; \item allows a mixture of positive and negative correlations. \end{inparaenum} Specifically, a compound symmetric correlation matrix is used for each block of size six at the bottom level. The correlation coefficient for each block is generated from a uniform distribution on the interval (0.2, 0.7). Algorithm 1 of \citet{haretal13} is used to impose correlations between different blocks. The resulting correlation matrix is shown in the second panel of Figure~\ref{fig:coefandcor}. Finally, the covariance matrix is constructed by sampling the standard deviations from a uniform distribution on the interval $(\sqrt{2}, \sqrt{6})$. Some of these covariances are converted into negatives to allow for a mixture of positive and negative correlations and is shown in the bottom panel of Figure~\ref{fig:coefandcor}.

\begin{figure}[!ht]
	\centering
	\includegraphics[width=0.55\textwidth, keepaspectratio=TRUE]{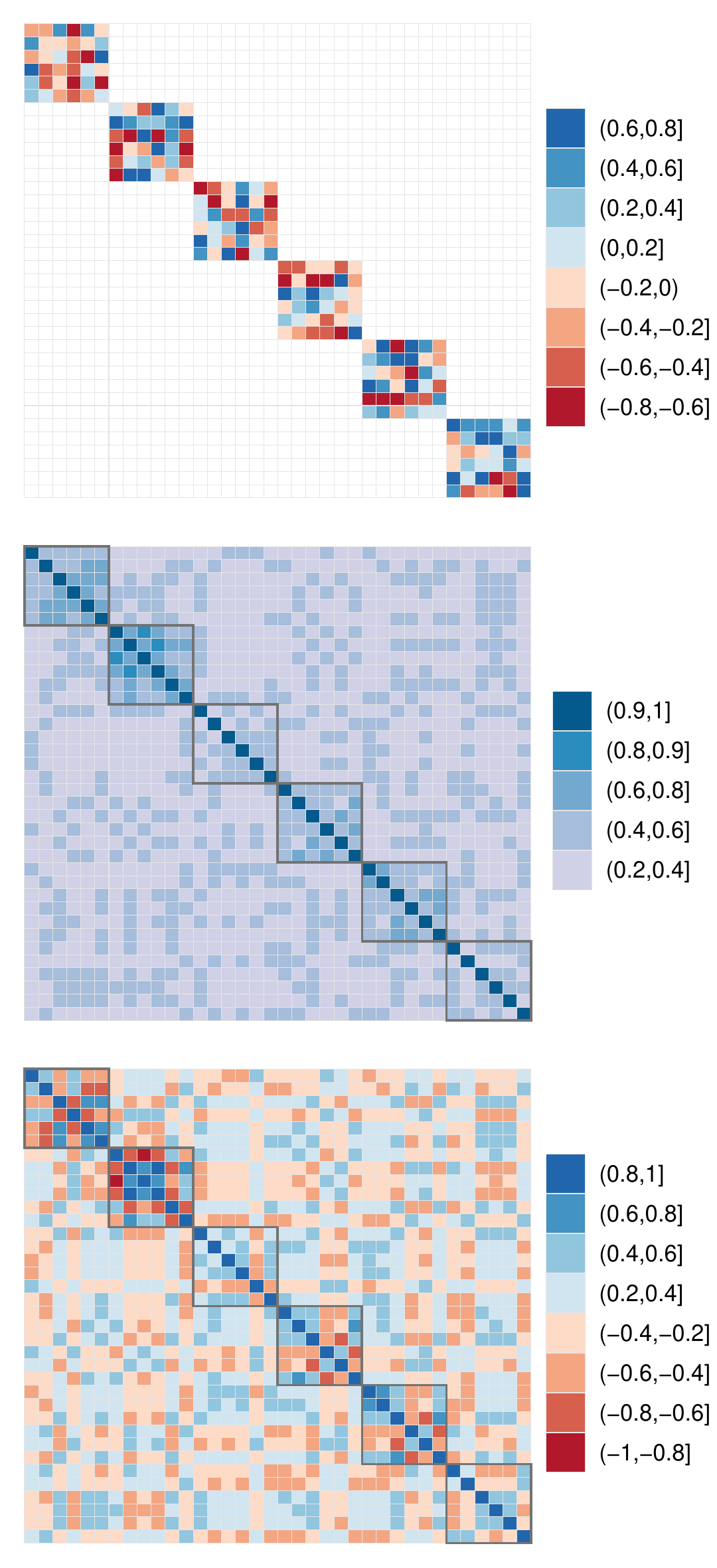}
	\caption{Coefficient matrix of the VAR(1) model (shown in the first panel), a matrix with non-negative correlations (shown in the second panel) and a matrix with positive and negative correlations (shown in the third panel).}
	\label{fig:coefandcor}
\end{figure}

For each series, $T = 101, 301$ or 501 observations are generated, with the last observation being withheld as a test set. Using the remaining observations as a training set, base forecasts are then computed from the best fitted ARMA models, which minimize AICc. These are then reconciled using different reconciliation methods. The process is repeated 1000 times.

The results for the non-negative, and a mixture of positive and negative error correlations are summarized in the left and right panels of Table~\ref{tbl:prial-large-structure}, respectively. Each entry in the table shows the percentage improvement in the MSE for the reconciled forecasts relative to the fitted/base values/forecasts. A negative (positive) entry shows a decrease (increase) in MSE of reconciled forecasts relative to that of the fitted/base values/forecasts. The bold entries identify the best performing methods.

The in-sample performance of forecast reconciliation approaches can be arranged in the following order based on the increasing MSE: EMinT-U, MinT(Sample), MinT(Shrink), WLS and OLS for the two choices of correlation matrices, all sample sizes and all the levels in the structure. For EMinT-U, MinT(Sample) and MinT(Shrink), the improvements are always negative. These observations are in accordance with the findings from Section~\ref{sec:simcor} and the theoretical results given in Section~\ref{sec:newtheory}. A similar pattern is observed in out-of-sample evaluations for $T = 301$ and 501. The out-of-sample performance of EMinT-U has decreased significantly when the non-negative error correlations are present and the sample size is 101. In contrast, MinT(Sample) and MinT(Shrink) perform the best. This behavior is also noted in Section~\ref{sec:simcor}. However, EMinT-U shows comparative results to MinT(Sample) and MinT(Shrink) when both positive and negative error correlations are present. Therefore, it can be noted that EMinT-U can perform poorly when all the correlations are positive, and it needs a larger sample size to perform well as expected.

\begin{landscape}
	\begin{table}[!htp]
		\caption{\label{tbl:prial-large-structure} Percentage relative improvements in MSE of forecast reconciliation methods for a larger structure.}
		\centering
		\fontsize{9}{12}\rm\tabcolsep=0.07cm
		\begin{tabular}{lrrrrrrrrrrrrrrrrrrr}
			\toprule
			& \multicolumn{9}{c}{Non-negative error correlations} & & \multicolumn{9}{c}{Positive and negative error correlations}\\
			\cmidrule{2-10} \cmidrule{12-20}
			& \multicolumn{4}{c}{In-sample} & & \multicolumn{4}{c}{Out-of-sample}  & & \multicolumn{4}{c}{In-sample} & & \multicolumn{4}{c}{Out-of-sample} \\
			\cmidrule{2-5} \cmidrule{7-10} \cmidrule{12-15} \cmidrule{17-20}
			& \multicolumn{1}{c}{Top} & \multicolumn{1}{c}{Level 1} & \multicolumn{1}{c}{Bottom} & \multicolumn{1}{c}{Overall} & & \multicolumn{1}{c}{Top} & \multicolumn{1}{c}{Level 1} & \multicolumn{1}{c}{Bottom} & \multicolumn{1}{c}{Overall} & & \multicolumn{1}{c}{Top} & \multicolumn{1}{c}{Level 1} & \multicolumn{1}{c}{Bottom} & \multicolumn{1}{c}{Overall} & & \multicolumn{1}{c}{Top} & \multicolumn{1}{c}{Level 1} & \multicolumn{1}{c}{Bottom} & \multicolumn{1}{c}{Overall} \\
			\cmidrule{2-5} \cmidrule{7-10} \cmidrule{12-15} \cmidrule{17-20}
			& \multicolumn{9}{c}{$T = 101$} & & \multicolumn{9}{c}{$T = 101$} \\
			\cmidrule{2-10} \cmidrule{12-20}
			BU & $-41.5$ & $44.9$ & $0.0$ & $-38.2$ & & $-18.2$ & $20.3$ & $0.0$ & $-4.6$ & & $-68.3$ & $13.1$ & $0.0$ & $-62.3$ & & $-46.5$ & $12.8$ & $0.0$ & $-17.4$ \\
			OLS & $-19.3$ & $14.8$ & $-7.4$ & $-18.0$ & & $-9.6$ & $6.3$ & $-3.7$ & $-4.3$ & & $-31.5$ & $6.9$ & $-1.5$ & $-28.7$ & & $-18.3$ & $7.4$ & $-0.8$ & $-6.4$ \\
			WLS & $-51.3$ & $9.8$ & $-10.1$ & $-48.9$ & & $-26.2$ & $4.8$ & $-4.8$ & $-14.0$ & & $-79.8$ & $-24.1$ & $-6.7$ & $-75.4$ & & $-57.6$ & $-8.1$ & $-3.5$ & $-28.8$ \\
			MinT(Sample) & $-62.3$ & $-18.0$ & $-40.0$ & $-60.6$ & & $-25.5$ & $3.4$ & $-11.1$ & $-15.1$ & & $-87.1$ & $-46.7$ & $-41.4$ & $-84.0$ & & $-61.7$ & $-12.5$ & $-12.1$ & $-34.3$ \\
			MinT(Shrink) & $-57.7$ & $-6.5$ & $-28.2$ & $-55.7$ & & $\pmb{-28.0}$ & $\pmb{0.9}$ & $\pmb{-11.2}$ & $\pmb{-17.3}$ & & $-85.4$ & $-38.9$ & $-31.0$ & $-81.7$ & & $\pmb{-62.5}$ & $-13.6$ & $-12.2$ & $-35.0$ \\
			EMinT-U & $\pmb{-87.9}$ & $\pmb{-77.1}$ & $\pmb{-89.9}$ & $\pmb{-87.5}$ & & $4.1$ & $32.3$ & $-11.1$ & $9.2$ & & $\pmb{-97.7}$ & $\pmb{-92.5}$ & $\pmb{-92.3}$ & $\pmb{-97.3}$ & & $-57.6$ & $\pmb{-19.0}$ & $\pmb{-20.8}$ & $\pmb{-36.7}$ \\
			\cmidrule{2-5} \cmidrule{7-10} \cmidrule{12-15} \cmidrule{17-20}
			& \multicolumn{9}{c}{$T = 301$} & & \multicolumn{9}{c}{$T = 301$} \\
			\cmidrule{2-10} \cmidrule{12-20}
			BU & $-38.1$ & $49.5$ & $0.0$ & $-34.7$ & & $-16.9$ & $24.6$ & $0.0$ & $-3.0$ & & $-69.5$ & $14.1$ & $0.0$ & $-63.4$ & & $-41.3$ & $7.8$ & $0.0$ & $-16.2$ \\
			OLS & $-18.0$ & $15.3$ & $-8.0$ & $-16.7$ & & $-9.0$ & $7.9$ & $-4.2$ & $-3.7$ & & $-31.5$ & $7.7$ & $-1.3$ & $-28.6$ & & $-17.0$ & $5.2$ & $-0.4$ & $-6.2$ \\
			WLS & $-47.9$ & $13.4$ & $-10.0$ & $-45.5$ & & $-25.4$ & $8.2$ & $-4.9$ & $-13.0$ & & $-79.9$ & $-22.8$ & $-6.5$ & $-75.4$ & & $-52.9$ & $-10.4$ & $-3.3$ & $-27.3$ \\
			MinT(Sample) & $-56.4$ & $-9.2$ & $-36.5$ & $-54.5$ & & $-29.4$ & $0.2$ & $-17.0$ & $-19.5$ & & $-86.2$ & $-41.9$ & $-38.0$ & $-82.8$ & & $-58.4$ & $-19.8$ & $-17.9$ & $-36.4$ \\
			MinT(Shrink) & $-55.1$ & $-5.9$ & $-32.6$ & $-53.2$ & & $-29.7$ & $0.4$ & $-15.9$ & $-19.3$ & & $-85.8$ & $-39.4$ & $-34.6$ & $-82.2$ & & $-58.5$ & $-19.1$ & $-17.0$ & $-36.0$ \\
			EMinT-U & $\pmb{-73.3}$ & $\pmb{-50.6}$ & $\pmb{-79.6}$ & $\pmb{-72.4}$ & & $\pmb{-30.1}$ & $\pmb{-5.1}$ & $\pmb{-40.7}$ & $\pmb{-25.3}$ & & $\pmb{-95.5}$ & $\pmb{-84.8}$ & $\pmb{-84.7}$ & $\pmb{-94.7}$ & & $\pmb{-71.3}$ & $\pmb{-47.4}$ & $\pmb{-47.3}$ & $\pmb{-58.0}$ \\
			\cmidrule{2-5} \cmidrule{7-10} \cmidrule{12-15} \cmidrule{17-20}
			& \multicolumn{9}{c}{$T = 501$} & & \multicolumn{9}{c}{$T = 501$}\\
			\cmidrule{2-10} \cmidrule{12-20}
			BU & $-37.6$ & $50.3$ & $0.0$ & $-34.2$ & & $-23.5$ & $22.4$ & $0.0$ & $-7.4$ & & $-69.7$ & $14.1$ & $0.0$ & $-63.5$ & & $-41.0$ & $8.4$ & $0.0$ & $-16.2$ \\
			OLS & $-17.8$ & $15.5$ & $-8.1$ & $-16.5$ & & $-9.6$ & $9.1$ & $-3.4$ & $-3.7$ & & $-31.4$ & $7.7$ & $-1.3$ & $-28.5$ & & $-16.6$ & $6.7$ & $-0.3$ & $-5.8$ \\
			WLS & $-47.4$ & $14.1$ & $-10.0$ & $-45.0$ & & $-29.4$ & $7.3$ & $-4.5$ & $-15.5$ & & $-79.8$ & $-22.6$ & $-6.5$ & $-75.3$ & & $-53.8$ & $-9.6$ & $-3.2$ & $-27.7$ \\
			MinT(Sample) & $-55.3$ & $-7.4$ & $-35.8$ & $-53.5$ & & $-35.5$ & $-1.5$ & $-17.5$ & $-23.5$ & & $-86.0$ & $-41.2$ & $-37.4$ & $-82.6$ & & $-60.9$ & $-19.5$ & $-18.7$ & $-37.9$ \\
			MinT(Shrink) & $-54.7$ & $-5.8$ & $-33.6$ & $-52.8$ & & $-35.0$ & $-1.3$ & $-16.8$ & $-23.0$ & & $-85.8$ & $-39.8$ & $-35.5$ & $-82.3$ & & $-61.1$ & $-19.2$ & $-17.9$ & $-37.7$ \\
			EMinT-U & $\pmb{-69.8}$ & $\pmb{-44.2}$ & $\pmb{-77.1}$ & $\pmb{-68.9}$ & & $\pmb{-36.8}$ & $\pmb{-10.1}$ & $\pmb{-43.6}$ & $\pmb{-31.0}$ & & $\pmb{-94.9}$ & $\pmb{-83.0}$ & $\pmb{-83.0}$ & $\pmb{-94.0}$ & & $\pmb{-74.4}$ & $\pmb{-50.0}$ & $\pmb{-50.0}$ & $\pmb{-61.0}$ \\
			\bottomrule
		\end{tabular}
	\end{table}
\end{landscape}

The left panels of Figure~\ref{fig:simlarind-1} show the percentage relative improvements in MSE for the in-sample reconciled forecasts relative to that for 1-step-ahead fitted values, and the right panels show the improvements from the out-of-sample reconciled forecasts relative to that for the base forecasts for each series in the structure when positive and negative error correlations are present. A negative (positive) entry shows a decrease (increase) in MSE of reconciled forecasts relative to that of the fitted/base values/forecasts. The series are arranged based on the performance of EMinT-U. We do not report the results for non-negative error correlations as the conclusions are qualitatively similar and available upon request.

As we noted in Section~\ref{sec:simcor}, the in-sample performance of EMinT-U is the greatest, which is then followed by MinT(Sample)/MinT(Shrink) and OLS for each series in the structure. Except for OLS, the improvements for other methods are always greater than the base forecasts. These patterns are not apparent in out-of-sample evaluations when $T = 101$. There are series for which EMinT-U forecasts are worst than the base forecasts. However, for such series, MinT forecasts have not performed poorly. As the sample size increases to $T = 501$, the in-sample patterns are also visible in out-of-sample evaluations.

\begin{figure}[!htp]
	\includegraphics[width=\textwidth]{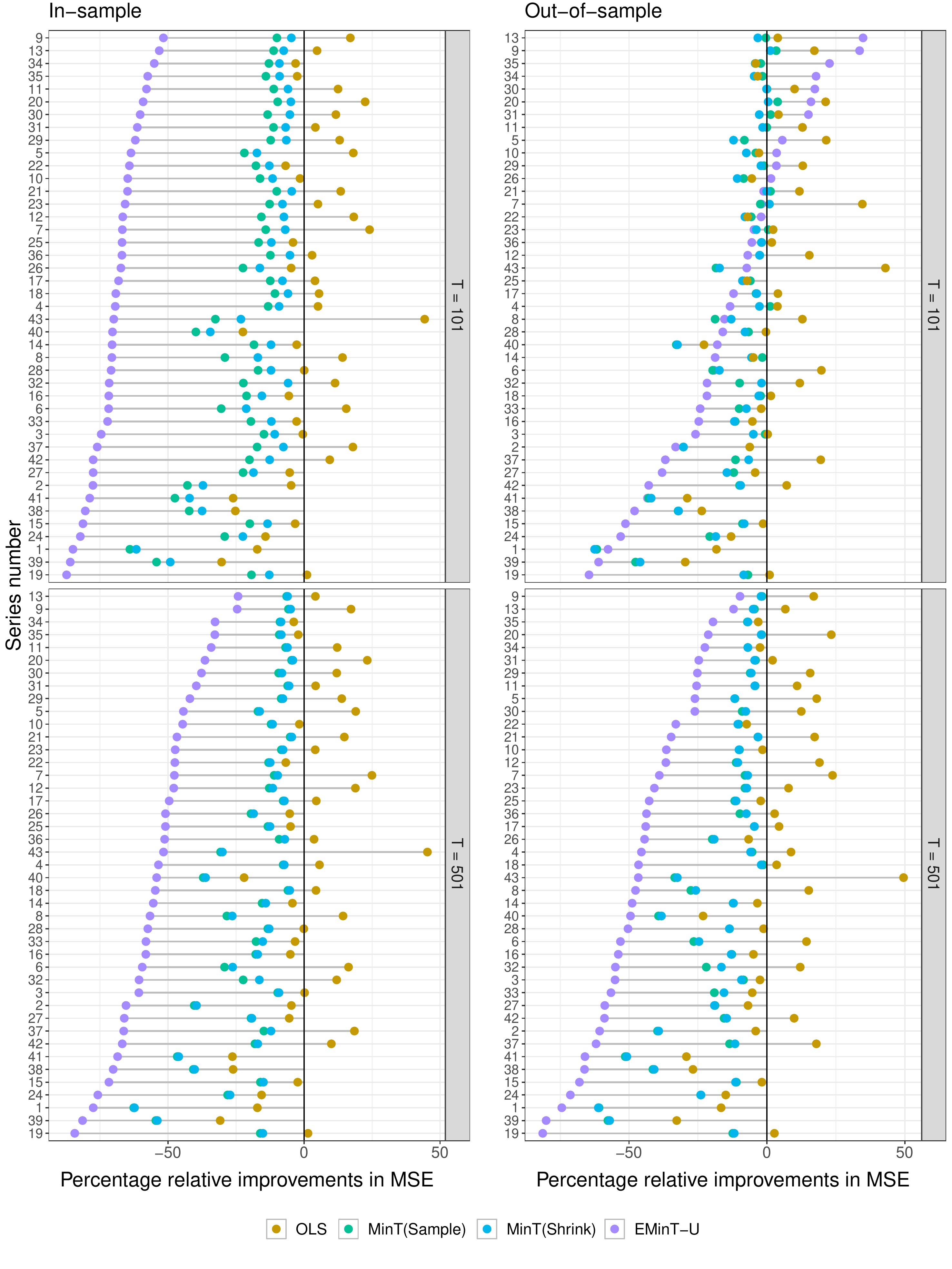}
	\caption{Percentage relative improvements in MSE for the in-sample (shown in the left panel) and out-of-sample (shown in the right panel) evaluations of each series when the contemporaneous error correlations are on the interval $(-1, 1)$. The series are sorted according to the performance of EMinT-U.}
	\label{fig:simlarind-1}
\end{figure}

\section{Application}
\label{sec:application}

For our empirical investigation, we consider Australian domestic tourism flows to build several hierarchical structures. We measure domestic tourism flow using ``visitor nights'', the total number of nights spent by Australians away from home. The data are managed by Tourism Research Australia and are collected through the national visitor survey conducted by computer-assisted telephone interviews. The information is gathered from an annual sample of 120,000 Australian residents aged 15 years or over. The data are monthly time series and span the period from January 1998 to December 2019.

Using the information available, we construct two simple hierarchies. The first hierarchy disaggregates the total visitor nights in Australia by the purpose of travel, whereas the second hierarchy disaggregates it by states and territories in Australia. There are four purposes of travel: holiday, visiting friends and relatives (VFR), business, and other, and seven states and territories: New South Wales (NSW), Victoria (VIC), Queensland (QLD), South Australia (SA), Western Australia (WA), Tasmania (TAS) and Northern Territory (NT).

The top panel of Figure~\ref{fig:vn-timeplots} shows the time plot of visitor nights in Australia. The middle and bottom panels show the time plots for four types of purpose of travel, and seven states and territories in Australia, respectively. We can see that, except for the `other' time series, the rest of the series show strong seasonal patterns. Most of the prominent series show diverse trends.

\begin{figure}
	\centering
	\includegraphics[width=0.9\linewidth]{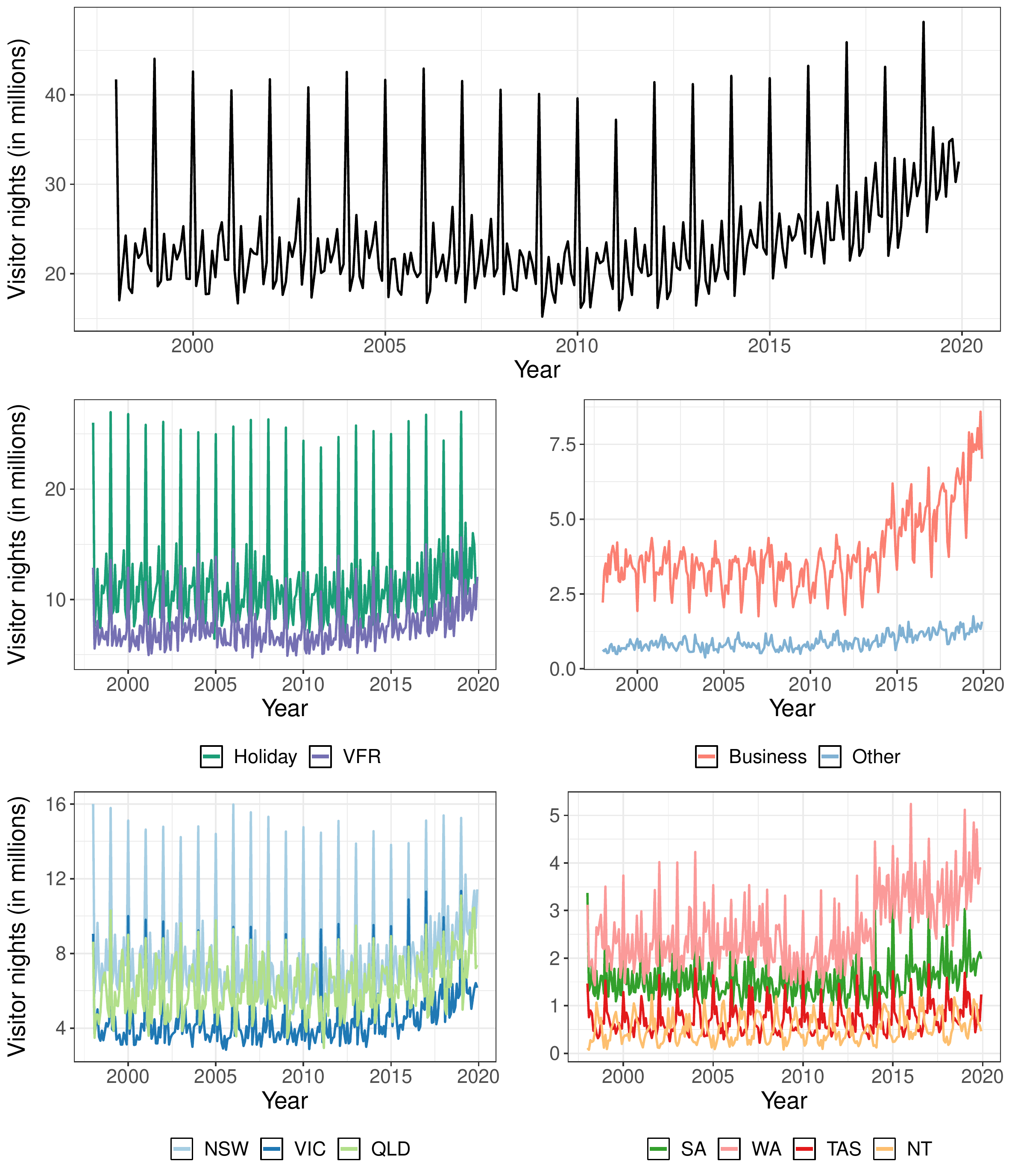}
	\caption{Time plots of visitor nights in Australia. The total number of visitor nights aggregated across purpose of travel or states and territories (shown in the top panel), total number of visitor nights for each of the 4 purposes of travel (shown in the middle panel) and each of the 7 states and territories (shown in the bottom panel).}
	\label{fig:vn-timeplots}
\end{figure}

For each series in each hierarchy we begin with a training set of size 120. In order to evaluate For each series in each hierarchy, we begin with a training set of size 120. To evaluate the forecast accuracy of EMinT-U, we need to ensure that the series in the structure are jointly weakly stationary. First, we remove the seasonal component from each series using the seasonal, trend and irregular decomposition using loess (STL). We assume that the seasonal component is periodic. For each seasonally adjusted series, we perform the Kwiatkowski–Phillips–Schmidt–Shin test to decide whether the series needs non-seasonal differencing or not. If at least one series in the hierarchy needs differencing, we apply the same differencing to all the series in the structure. This ensures that the transformed series satisfy the same aggregation constraints as the original data. For each transformed series, the best fitted seasonal ARMA model is identified by minimizing AICc. Then the base forecasts are produced for 1-step-ahead for each series in the structure. Then the base forecasts are reconciled using the alternative approaches. We roll the training window forward by one observation until November 2019.

\begin{table}[ht]
	\centering
	\caption{Percentage relative improvements in MSE for forecast reconciliation methods for the Australian domestic tourism data set.}
	\label{tbl:prial-1-level}
	\fontsize{9}{12}\rm\tabcolsep=0.05cm
	\begin{tabular}{lrrrrrrrrrrrrrrr}
		\toprule
		& \multicolumn{7}{c}{Hierarchy 1} & & \multicolumn{7}{c}{Hierarchy 2} \\
		\cmidrule{2-8} \cmidrule{10-16}
		& \multicolumn{3}{c}{In-sample} & & \multicolumn{3}{c}{Out-of-sample} & & \multicolumn{3}{c}{In-sample} & & \multicolumn{3}{c}{Out-of-sample} \\
		\cmidrule{2-4} \cmidrule{6-8} \cmidrule{10-12} \cmidrule{14-16}
		& Australia & Purpose & Overall & & Australia & Purpose & Overall & & Australia & States & Overall & & Australia & States & Overall \\
		\cmidrule{2-4} \cmidrule{6-8} \cmidrule{10-12} \cmidrule{14-16}
		BU & $-2.7$ & $0.0$ & $-1.6$ & & $-9.6$ & $0.0$ & $-6.0$ & & $-3.0$ & $0.0$ & $-2.0$ & & $\pmb{-3.0}$ & $\pmb{0.0}$ & $\pmb{-2.0}$\\
		OLS & $-1.4$ & $0.5$ & $-0.7$ & & $-2.5$ & $2.9$ & $-0.5$ & & $-1.0$ & $0.5$ & $-0.5$ & & $-0.9$ & $0.6$ & $-0.4$ \\
		WLS & $-3.1$ & $0.0$ & $-1.8$ & & $-6.5$ & $2.6$ & $-3.1$ & & $-3.4$ & $-0.2$ & $-2.3$ & & $-2.9$ & $\pmb{0.0}$ & $-1.9$ \\
		MinT(Sample) & $-4.6$ & $-1.4$ & $-3.3$ & & $-5.2$ & $2.2$ & $-2.4$ & & $-4.9$ & $-1.5$ & $-3.7$ & & $2.2$ & $3.3$ & $2.5$ \\
		MinT(Shrink) & $-3.8$ & $-0.6$ & $-2.5$ & & $-6.3$ & $2.3$ & $-3.0$ & & $-4.1$ & $-0.7$ & $-2.9$ & & $-1.7$ & $0.7$ & $-0.9$ \\
		EMinT-U & $\pmb{-7.9}$ & $\pmb{-4.7}$ & $\pmb{-6.6}$ & & $\pmb{-10.0}$ & $\pmb{-0.4}$ & $\pmb{-6.4}$ & & $\pmb{-10.0}$ & $\pmb{-6.7}$ & $\pmb{-8.9}$ & & $5.5$ & $6.7$ & $5.9$ \\
		\bottomrule
	\end{tabular}
\end{table}

Table~\ref{tbl:prial-1-level} presents the results of the rolling window forecast evaluation. The left panel shows in-sample and out-of-sample percentage relative improvements in MSE relative to that of the base forecasts for the first hierarchy (disaggregated by the purpose of travel), and the right panel shows that for the second hierarchy (disaggregated by states and territories). A negative (positive) entry shows a decrease (increase) in MSE relative to the base forecasts. The bold entries identify the best performing methods. The in-sample performance of EMinT-U is the best for both the hierarchies, which is followed by MinT(Sample), MinT(Shrink), WLS and OLS (in the increasing order of MSE) for all the levels in the structure. The out-of-sample performance of hierarchy 1 is greatest for EMinT-U, and BU is the second best. One reason for BU to perform well could be the high signal-to-noise ratio present at the bottom level. For hierarchy 2, BU marginally outperforms WLS and could be due to the same reason noted above. The performance of EMinT-U is worst than MinT(Sample).

Tables~\ref{tbl:in-hier1} and \ref{tbl:in-hier2} summarizes the percentage relative improvements in MSE for forecast reconciliation methods relative to that of the base forecasts for each series in hierarchy 1 and 2, respectively. The analysis is performed separately for in-sample (shown in the left panel) and out-of-sample (shown in the right panel). The in-sample evaluations revealed that EMinT-U is the best and MinT(Sample) is the second best for all the series in both structures. The out-of-sample evaluation for hierarchy 1 has only three series out of five with the best performances, whereas none of the series in hierarchy 2 show improvements for EMinT-U. The decrease in performance for EMinT-U can be due to several reasons: \begin{inparaenum}[(a)]\item EMinT-U estimates more parameters than MinT(Sample) (i.e., $n \times m = 56$ (EMinT-U) as opposed to $n \times m^* = 7$ (MinT)); \item for positive error correlations we need a larger sample size to show better performances \end{inparaenum}. Figure~\ref{fig:correlation-states} shows the 1-step-ahead in-sample correlation matrix from the last iteration of the rolling window forecast evaluation. We can observe that most of the correlations are positive and varies from weak to moderate in strength. Therefore, we can suspect that the positive error correlation might be one of the factors for the worst performance in EMinT-U. The first point that we noted above was also observed by \citet{benkoo19}, and they proposed to consider regularized approaches to reduce the number of parameters that need to be estimated.

\begin{table}[ht]
	\centering
	\caption{Percentage relative improvements in MSE for each series in hierarchy 1.}
	\label{tbl:in-hier1}
	\fontsize{9}{12}\rm\tabcolsep=0.07cm
	\begin{tabular}{lrrrrrrrrrrr}
		\toprule
		& \multicolumn{5}{c}{In-sample} & & \multicolumn{5}{c}{Out-of-sample}\\
		\cmidrule{2-6} \cmidrule{8-12}
		& Australia & Holiday & VFR & Business & Other & & Australia & Holiday & VFR & Business & Other \\
		\cmidrule{2-6} \cmidrule{8-12}
		BU & $-2.7$ & 0.0 & 0.0 & 0.0 & 0.0 & & $-9.6$ & 0.0 & \pmb{0.0} & 0.0 & \pmb{0.0} \\
		OLS & $-1.4$ & 0.2 & 0.2 & $-0.2$ & 15.8 & & $-2.5$ & 2.9 & 3.9 & $-0.9$ & 21.3 \\
		WLS & $-3.1$ & 0.3 & $-0.2$ & $-0.5$ & 0.0 & & $-6.5$ & 3.4 & 2.6 & $-0.5$ & 0.4 \\
		MinT(Sample) & $-4.6$ & $-1.4$ & $-1.3$ & $-1.2$ & $-1.9$ & & $-5.2$ & 1.3 & 6.1 & $-1.1$ & 3.4 \\
		MinT(Shrink) & $-3.8$ & $-0.5$ & $-0.8$ & $-0.9$ & $-0.9$ & & $-6.3$ & 2.6 & 3.5 & $-0.8$ & 1.0 \\
		EMinT-U & $\pmb{-7.9}$ & $\pmb{-4.0}$ & $\pmb{-5.7}$ & $\pmb{-5.9}$ & $\pmb{-5.0}$ & & $\pmb{-10.0}$ & $\pmb{-0.5}$ & 1.7 & $\pmb{-4.2}$ & 6.7 \\
		\bottomrule
	\end{tabular}
\end{table}

\begin{table}[ht]
	\centering
	\caption{Percentage relative improvements in MSE for each series in hierarchy 2.}
	\label{tbl:in-hier2}
	\fontsize{9}{12}\rm\tabcolsep=0.07cm
	\begin{tabular}{lrrrrrrrrrrrrrrrrr}
		\toprule
		& \multicolumn{8}{c}{In-sample} & & \multicolumn{8}{c}{Out-of-sample}\\
		\cmidrule{2-9} \cmidrule{11-18}
		& Australia & NSW & VIC & QLD & SA & WA & TAS & NT & & Australia & NSW & VIC & QLD & SA & WA & TAS & NT \\
		\cmidrule{2-9} \cmidrule{11-18}
		BU & $-3.0$ & 0.0 & 0.0 & 0.0 & 0.0 & 0.0 & 0.0 & 0.0 & & $\pmb{-3.0}$ & \pmb{0.0} & \pmb{0.0} & 0.0 & \pmb{0.0} & 0.0 & \pmb{0.0} & 0.0 \\
		OLS & $-1.0$ & 0.3 & 1.1 & $-0.5$ & 7.3 & $-0.5$ & 6.1 & 7.7 & & $-0.9$ & 0.3 & 0.6 & $-0.2$ & 5.7 & 0.2 & 12.5 & 5.5 \\
		WLS & $-3.4$ & 0.1 & 0.0 & $-0.4$ & 0.2 & $-0.5$ & $-0.2$ & 0.0 & & $-2.9$ & 0.4 & \pmb{0.0} & $\pmb{-0.3}$ & 0.1 & $\pmb{-0.2}$ & 0.2 & $\pmb{-0.1}$ \\
		MinT(Sample) & $-4.9$ & $-1.0$ & $-2.8$ & $-1.5$ & $-1.5$ & $-1.5$ & $-1.2$ & $-0.6$ & & 2.2 & 4.3 & 5.5 & 2.1 & 3.8 & 1.6 & 0.8 & 2.3 \\
		MinT(Shrink) & $-4.1$ & $-0.3$ & $-1.2$ & $-0.9$ & $-0.5$ & $-1.0$ & $-0.6$ & $-0.2$ & & $-1.7$ & 1.1 & 1.2 & 0.3 & 0.7 & 0.0 & 0.2 & 0.3 \\
		EMinT-U & $\pmb{-10.0}$ & $\pmb{-6.7}$ & $\pmb{-10.2}$ & $\pmb{-6.1}$ & $\pmb{-7.2}$ & $\pmb{-4.5}$ & $\pmb{-6.9}$ & $\pmb{-5.8}$ & & 5.5 & 10.5 & 6.3 & 3.3 & 5.0 & 5.7 & 12.3 & 10.5 \\
		\bottomrule
	\end{tabular}
\end{table}

\begin{figure}[h]
	\centering
	\includegraphics[width=0.54\linewidth]{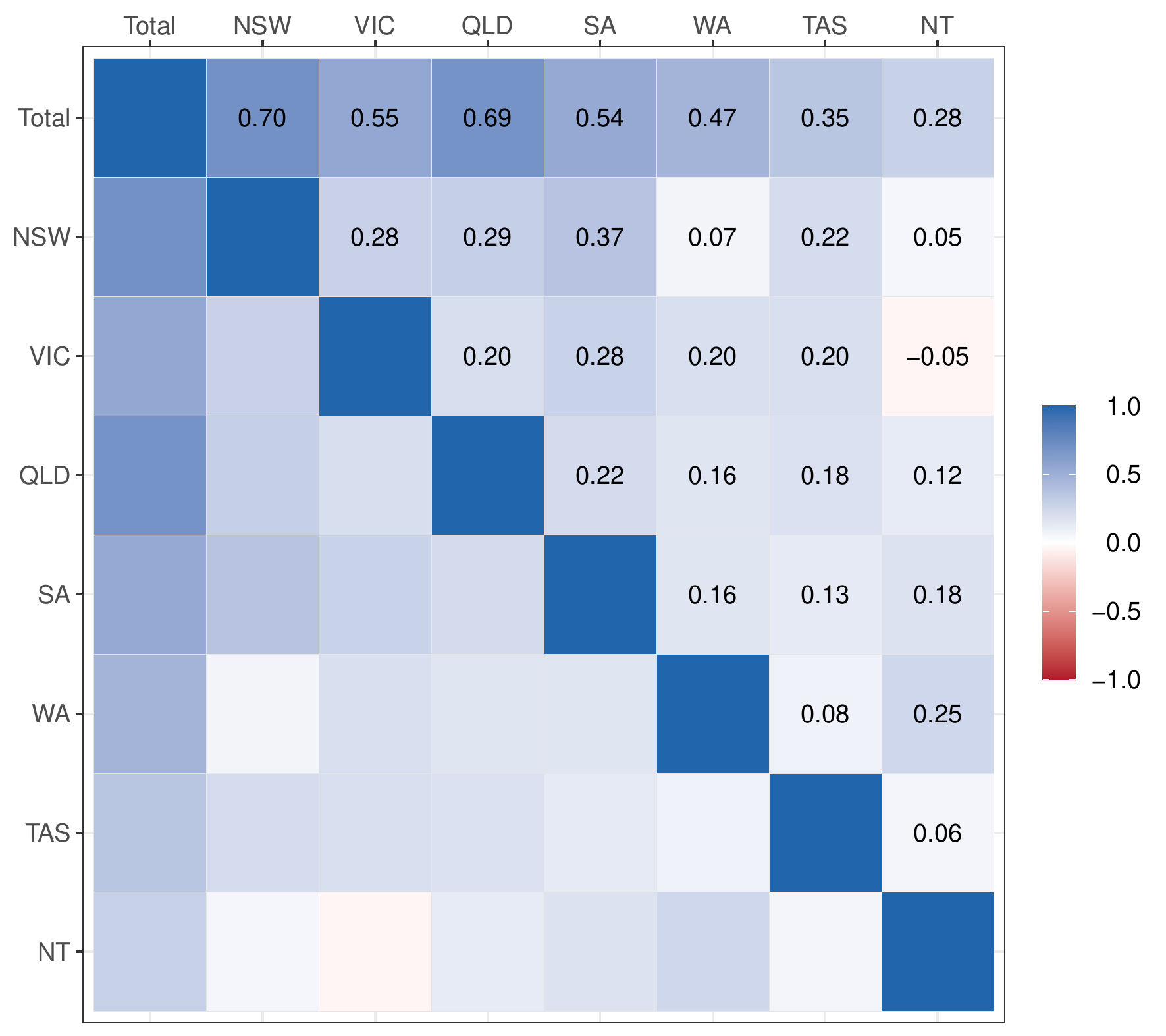}
	\caption{Heatmap of in-sample 1-step-ahead base forecast error correlation matrix of hierarchy 2.}
	\label{fig:correlation-states}
\end{figure}

\section{Conclusion}
\label{sec:conclusion}
This paper aimed to study the properties of point forecast reconciliation methods as there is a growing interest in using these methods among statistical and general scientific communities. A paradigm shift in point forecast reconciliation occurred after the work of \citet{hynetal11}. They introduced a method called GLS, which needs an estimate of the covariance matrix of the coherence errors. They avoided the estimation of the covariance matrix by using the OLS method. Recently, \citet{wicetal19} proposed an alternative method, MinT, which has the same form as GLS, but it needs an estimate of the covariance matrix of the base forecast errors. They noted that the latter covariance matrix could be estimated in practice while the former is not. In this study, we proved that even though these two methods minimize different loss functions and involve different covariance matrices in the final expressions, they both lead to the same solution. We also theoretically showed that, on average, MinT reconciled forecasts improve upon base forecasts (lowest total MSE), and the mean squared error of each series in the structure for MinT is smaller than that for either OLS or base.

We proposed a reconciliation method (MinT-U) similar to \citet{benkoo19} by deviating from the projection matrices. We showed that this method could perform better than MinT, on average, and the mean squared error of each series after applying this method is smaller than that for MinT. Even though this result is promising, the applicability of the method is limited to jointly weakly stationary time series. Despite these restrictions, this method will provide a foundation for researchers to focus also on matrices that are not necessarily projections. We evaluated these methods using simulated and real data. The performance of EMinT-U is impacted by small sample sizes, and hence regularization methods can be applied to reduce the number of parameters that need to be estimated. We leave this to be addressed in a future paper.

\section*{Acknowledgement}
The author greatly appreciates valuable comments and insights from Professor Rob J Hyndman, Professor Thomas Lumley, Associate Professor Ilze Ziedins and Dr\@. Ciprian Giurcaneanu. The author wishes to acknowledge the use of the New Zealand eScience Infrastructure (NeSI) high-performance computing facilities as part of this research. New Zealand's national facilities are provided by NeSI and funded jointly by NeSI's collaborator institutions and through the Ministry of Business, Innovation \& Employment's Research Infrastructure programme. URL \url{https://www.nesi.org.nz}.

\newpage

\printbibliography

\end{document}